\documentclass[11pt]{article}
\usepackage{graphicx,amsmath,amsbsy,amssymb,citesort,epsfig,comment}
\usepackage{epstopdf}
\usepackage{amsthm}
\usepackage{bibspacing}
\def\real    { \mathbb{R} }

\newtheorem{thm}{Theorem}[section]
\newtheorem{lemma}{Lemma}[section]

\newcommand{\bitem}{\begin{itemize}}
\newcommand{\eitem}{\end{itemize}}

\newcommand{\argmin}{\mathrm{argmin}}

\newcommand{\supp}{\mathrm{supp}}
\newcommand{\beqn}{\begin{equation}}
\newcommand{\eeqn}{\end{equation}}
\newcommand{\balign}{\begin{align}}
\newcommand{\ealign}{\end{align}}

\newcommand{\ambient}{B}
\newcommand{\sparsity}{W}
\newcommand{\meas}{M}
\newcommand{\subsamp}{\rho}
\newcommand{\subsampmax}{\subsamp_{\mathrm{max}}}
\newcommand{\subsampcs}{\subsamp_{\mathrm{cs}}}

\newcommand{\betamax}{\beta_C^{\mathrm{max}}}
\newcommand{\betamin}{\beta_C^{\mathrm{min}}}

\newcommand{\balpha}{\boldsymbol{\alpha}}

\newcommand{\bSigma}{\boldsymbol{\Sigma}}
\newcommand{\y}{\boldsymbol{y}}
\newcommand{\e}{\boldsymbol{e}}
\newcommand{\x}{\boldsymbol{x}}
\newcommand{\n}{\boldsymbol{n}}

\newcommand{\I}{\boldsymbol{I}}
\newcommand{\bPhi}{\boldsymbol{\Phi}}
\newcommand{\bPsi}{\boldsymbol{\Psi}}
\newcommand{\R}{\boldsymbol{R}}
\newcommand{\A}{\boldsymbol{A}}
\newcommand{\B}{\boldsymbol{B}}
\newcommand{\U}{\boldsymbol{U}}
\newcommand{\V}{\boldsymbol{V}}
\newcommand{\sN}{\mathcal{N}}

\newcommand{\norm}[2][2]{\left\| #2 \right\|_{#1}} 
\newcommand{\iprod}[2]{\left\langle {#1},{#2} \right\rangle}
\newcommand{\expec}[1]{\mathbb{E}\left( #1 \right)}
\newcommand{\prob}[1]{\mathbb{P}\left( #1 \right)}
\newcommand{\quant}[1]{Q_b \left( #1 \right)}
\newcommand{\trace}[1]{\mathrm{Tr} \left( #1 \right)}

\newlength{\imgwidth}
\allowdisplaybreaks[1]

\title{The Pros and Cons of Compressive Sensing \\ for Wideband Signal Acquisition: \\ Noise Folding vs.\ Dynamic Range}

\author{Mark~A.~Davenport, Jason~N.~Laska, John~R.~Treichler, and~Richard~G.~Baraniuk\thanks{M.D. is with the Department of Statistics, Stanford University, Stanford, CA.  J.L. is with Dropcam, Inc., San Francisco, CA.  J.T. is with Applied Signal Technology, Inc., Sunnyvale, CA.  R.B. is with the Department of Electrical and Computer Engineering, Rice University, Houston, TX. M.D. was supported by the grant NSF DMS-1004718.  J.L. and R.B. were supported by the grants NSF CCF-0431150, CCF-0728867, CCF-0926127, CNS-0435425, and CNS-0520280, DARPA/ONR N66001-08-1-2065, N66001-11-1-4090, ONR N00014-07-1-0936, N00014-08-1-1067, N00014-08-1-1112, and N00014-08-1-1066, AFOSR FA9550-07-1-0301 and FA9550-09-1-0432, ARO MURI W911NF-07-1-0185 and W911NF-09-1-0383, and the Texas Instruments Leadership University Program.  This work extends preliminary results presented in~\cite{TreicDB_Application}. For many insightful discussions, thanks to Glenn Collins, Jeffrey Harp, and Jared Sorensen at AST, and Stephen Schnelle, J.P.\ Slavinsky, and Andrew Waters at Rice.}}

\begin{document}
\maketitle

\begin{abstract}
\noindent Compressive sensing (CS) exploits the sparsity present in many signals to reduce the number of measurements needed for digital acquisition. With this reduction would come, in theory, commensurate reductions in the size, weight, power consumption, and/or monetary cost of both signal sensors and any associated communication links. This paper examines the use of CS in the design of a wideband radio receiver in a noisy environment. We formulate the problem statement for such a receiver and establish a reasonable set of requirements that a receiver should meet to be practically useful. We then evaluate the performance of a CS-based receiver in two ways: via a theoretical analysis of its expected performance, with a particular emphasis on noise and dynamic range, and via simulations that compare the CS receiver against the performance expected from a conventional implementation.  On the one hand, we show that CS-based systems that aim to reduce the number of acquired measurements are somewhat sensitive to signal noise, exhibiting a $3$dB SNR loss per octave of subsampling, which parallels the classic noise-folding phenomenon.  On the other hand, we demonstrate that since they sample at a lower rate, CS-based systems can potentially attain a significantly larger dynamic range.  Hence, we conclude that while a CS-based system has inherent limitations that do impose some restrictions on its potential applications, it also has attributes that make it highly desirable in a number of important practical settings.
\end{abstract}

\section{Introduction}
\label{sec:intro}

The recently developed \emph{compressive sensing} (CS) framework has radically transformed our understanding of signal acquisition~\cite{Baran_Compressive,Cande_Compressive,DavenDEK_Introduction,Donoh_Compressed}. Rather than sampling signals at twice their bandwidth, CS liberates us from the yoke of Shannon-Nyquist theory by proposing that signals be sampled at a rate proportional to their information content.  Specifically, CS enables the acquisition of \emph{sparse} signals, i.e., those that have only a few non-zero coefficients in some suitable representation.

Despite the attention that CS has received in recent years, several key facets have remained unaddressed when it comes to using CS in practice.  For example, there is rich analysis for bounded noise or Gaussian noise on CS measurements; however, in practice noise is typically present in both the measurements and the signal itself.  Moreover, the impact of finite-range amplitude quantization has been largely ignored.

In this paper we tackle these issues by examining the nuts and bolts of a specific application: the CS acquisition receiver.  We begin with a detailed ``engineer's guide'' to CS, emphasizing the practical requirements.  We then rigorously analyze the setting where noise is present not only in the measurements but also in the signal.  In this context we study the impact of \emph{noise folding} owing to subsampling in the presence of signal noise.  On the negative side, we show that when noise is added to the input signal, a reduction in the sampling rate by one half induces a $3$dB SNR loss on the final signal estimate~\cite{TreicDB_Application}.  This effect, mainly overlooked in the past, highlights a real cost associated with the use of CS.  We also note that while we focus here primarily on the problem of using CS to acquire noisy wideband signals, our analysis is general and applies to any potential application of CS.

To balance the negatives of noise folding, we also study the {\em dynamic range} of a CS-based receiver.  On the positive side, we demonstrate that the dynamic range of a CS-based receiver can be significantly improved compared to a conventional analog-to-digital converter (ADC), primarily since by sampling at a lower rate we can typically quantize with a higher effective number of bits.  This theoretical analysis is confirmed by simulations that demonstrate that the impact of noise and finite dynamic range quantization follow our predictions.  Empirically, we observe that each reduction in the sampling rate by one half leads to an SNR gain of approximately $5$dB.

This paper is organized as follows.  Section~\ref{sec:reqs} formulates the problem statement and establishes a set of requirements that a receiver should meet to be highly attractive for practical use. Section~\ref{sec:csthy} reviews the CS theory, and Section~\ref{sec:whitenoise} builds on this theory to analyze the performance of a CS-based receiver in the presence of signal and measurement noise.  Section~\ref{sec:dr} provides an analysis of the dynamic range of a CS-based receiver, demonstrating the potential for substantial improvement.  Section~\ref{sec:tests} presents the results of a set of simulations that compare a CS-based receiver with our theoretical predictions.  These results are further discussed in Section~\ref{sec:sims}. Recommendations for additional study and investigation appear in Section~\ref{sec:conc}.

\section{Putative Requirements}
\label{sec:reqs}

Our objective in this paper is to explore the attributes and capabilities of CS by examining how it might be applied to meet a specific set of requirements. The particular application addressed is a wideband radio frequency (RF) signal acquisition receiver, a device commonly used in both commercial and military systems to monitor a wide band of radio frequencies for the purposes of ({\em i}) detecting the presence of signals, ({\em ii})~characterizing them, and, where appropriate, ({\em iii}) extracting a specific signal from the several that might be present within that band. Many types of acquisition receivers have been designed, built, and sold over the years, but we will choose here a set of putative requirements for such a receiver to ease comparisons and analysis. The reader is invited to repeat the comparison for other parameter choices.

\begin{table}[t]
\caption{A putative set of specifications for an advanced RF signal acquisition receiver.}
\label{tab:reqs}
\centering
\begin{tabular}{lcc}
\hline\hline
Instantaneous bandwidth & $B/2$ & 500 MHz \\
Instantaneous dynamic range & $D\!R$ & $96$dB \\
SNR degradation/noise figure & $N\!F$ & $12$dB \\
Maximum signal bandwidth & $W/2$ & 200 kHz \\
\hline\hline
\end{tabular}
\end{table}

The attributes that characterize an acquisition receiver typically fall into two categories: technical specifications --- such as instantaneous bandwidth --- and various ``costs'' --- such as size, weight, and power consumption (SWAP) and monetary cost. In this paper we will address only the few most important technical specifications:
\begin{itemize}

\item {\em Instantaneous bandwidth} --- the radio frequency (RF) range over which signals will be accepted by the receiver and handled with full fidelity.

\item {\em Instantaneous dynamic range} --- the ratio of the maximum to minimum signal power level with which received signals can be handled with full fidelity.

\item {\em SNR degradation} --- sometimes termed ``noise figure,'' a measure of the tendency of the receiver to lower the input signal-to-noise ratio (SNR) of a received signal, usually measured in dB. The root cause of this degradation has historically depended on the technology used to build the receiver.

\item {\em Maximum signal bandwidth} --- the maximum combined bandwidth of the constituent signals in the acquisition bandwidth of the receiver.

\end{itemize}
These requirements must be met subject to many constraints, including, at least, SWAP and monetary cost. There are also typically system-level constraints, such as the bandwidth available for communicating what the receiver has discovered to other assets or a central processing facility.

Historically RF signal acquisition receivers were first built using purely analog technology, then, more recently, with analog technology conditioning the signal environment sufficiently to employ a high-rate ADC followed by digital processing, storage, and/or transmission. If and when it can be applied, CS offers the promise to ({\em i}) increase the instantaneous input bandwidth, ({\em ii}) lower all of the cost attributes of the sensor, and ({\em iii}) move the computationally intensive portions of the acquisition process away from the sensor and toward a central processing facility.

For the purposes of the comparisons to be made in this paper, we will assume a set of requirements for an acquisition system that are rather audacious and would at the least stress conventional implementations at the present time.  These requirements are listed in Table~\ref{tab:reqs}.  To meet the bandwidth and dynamic range requirements, conventional designs would typically be forced to use techniques based on scanning narrowband receivers across the band.  If CS-based systems can be shown to work in such settings without the need for scanning at the receiver, then they would have broad application.

In order to apply CS, we must make two  last, but important, assumptions:
\begin{enumerate}

\item {\em Signal sparsity} --- In order to meet the first-order assumption of all CS techniques, in this paper we assume that the input signal is sparse.  To be concrete, in Table~\ref{tab:reqs} we assume that the sum of the bandwidths of all signals present in the full acquisition band is no more than 200 kHz. Note that this is significantly smaller than the instantaneous bandwidth of 500 MHz.  Thus we are assuming that the RF input to the receiver is significantly sparse in the frequency domain (the instantaneous bandwidth is only $1/2500$ occupied). Although inputs with this level of spectral sparsity are not common, they exist often enough to make a solution useful if it can be found. To test the impact of the sparsity assumption for this application, we will evaluate the performance, both theoretically and in simulation, for both the case where the input is noise-free, so that the input signal is truly sparse, and in the more practical case where the input is contaminated with additive white noise.

\item {\em Processing asymmetry} --- Our objective is to minimize all receiver and data link costs,  i.e.,  the SWAP and monetary cost of the receiver and the bandwidth required for transmission. We assume that once data is acquired and transmitted, we are prepared to invest heavily in a (possibly centralized) system that can do as much processing as needed to detect, characterize, and/or recover the signal of interest.  In other words, we assume that there is no cost to processing the receiver output, while there is high cost to the receiver acquisition and data forwarding processes.

\end{enumerate}

\section{Compressive Sensing Background}
\label{sec:csthy}

\subsection{CS signal acquisition}
\label{ssec:csthy}

In the present work we are concerned with the acquisition of a real-valued, continuous-time signal of interest, which we will denote by $x(t)$.  We will assume that $x(t)$ is bandlimited with instantaneous bandwidth $\ambient/2$ Hz.  We will further assume that $x(t)$ is sparse in the sense that the total occupied bandwidth, denoted by $\sparsity/2$, is much smaller than the instantaneous bandwidth, i.e., $\sparsity \ll \ambient$.  From the Shannon-Nyquist sampling theorem, we have that $x(t)$ can be perfectly reconstructed from a set of uniform samples taken at a rate of $\ambient$ Hz, and thus we can equivalently consider the problem of acquiring the corresponding discrete-time version of our signal of interest, which we will denote by $x[n]$.

In both the CS theory and in practice, we will actually be concerned with acquiring a finite window of $x(t)$, or equivalently $x[n]$.  Without loss of generality, we will assume that we are interested in a window of duration $T=1$ seconds.\footnote{We set $T=1$ for convenience; for other $T$ one can merely replace $\ambient$ and $\sparsity$ below with $\ambient T$ and $\sparsity T$, {\em mutatis mutandis}.} In this paper we will follow~\cite{TroppLDRB_Beyond} and make the additional simplifying assumption that $x(t)$ has a finite number of bounded harmonics, i.e.,
$$
x(t) = \bPsi (\balpha) = \sum_{k=0}^{\ambient-1} \alpha_k \psi_k(t),
$$
where $\psi_k(t) = e^{j 2 \pi k t}$ are the Fourier basis functions and where $\balpha = [\alpha_0, \alpha_1, \ldots, \alpha_{\ambient-1}]$ is a complex-valued vector of length $\ambient$. We will further assume that $\balpha$ has exactly $\sparsity$ nonzeros, corresponding to the active frequencies in the acquisition bandwidth. Thus, we assume that $\alpha_k = 0$ exactly outside of these $\sparsity$ nonzeros, in which case we say that the vector $\balpha$ is $\sparsity$-{\em sparse}. In practice, we do not observe truly $\sparsity$-sparse signals for two main reasons: ({\em i}) the spectrum will typically be contaminated by noise, and ({\em ii}) real-world signals are never perfectly bandlimited, and moreover restricting our attention to a finite-length window leads to an inevitable ``leakage'' of energy outside the signal's original support.  Although there exist alternative approaches to help mitigate the effect of windowing that build on classical spectral estimation techniques~\cite{DuartB_Spectral,DavenW_Compressive}, an analysis of these more specialized approaches is beyond the scope of this paper and is largely irrelevant to our discussion.  Hence, we will assume throughout this paper that noise dominates the leakage and that $\balpha$ can be modeled as a Fourier-sparse vector contaminated by noise.

The Shannon-Nyquist sampling theorem states that $\ambient$ samples of $x(t)$ per $T=1$ second contain all of the information in $x(t)$. Our goal in CS is to do better: to acquire $x(t)$ via $\meas = \ambient/\subsamp$ measurements per second with $\subsamp \ge 1$ as large as possible. The {\em subsampling factor} $\subsamp$ strongly affects the various costs (i.e., SWAP and monetary cost) described in Section~\ref{sec:reqs}. Observe that if the locations of the $\sparsity$ nonzeros of $\balpha$ are known {\em a priori}, then by filtering and decimation we could drive $\subsamp$ as large as $\subsampmax = \ambient/\sparsity$.  Our aim is to show that we can actually acquire $x(t)$ via a set of $\meas$ nonadaptive, linear measurements that require no {\em a priori} knowledge of the locations of the nonzeros of $\balpha$, with $\subsamp$ nearly as large as $\subsampmax$.

Towards this end, we acquire the measurements
\begin{equation} \label{eq:phi1}
\y = \bPhi(x(t)) + \e,
\end{equation}
where $\bPhi$ is a linear {\em measurement operator} that maps functions defined on $[0,1]$ to a length $\meas$ vector $\y$ of measurements, and where $\e$ is a length $\meas$ vector that represents {\em measurement noise} generated by the acquisition hardware. The central theoretical question in CS is how to design $\bPhi$ to ensure that we will be able to recover $x(t)$ from the measurements $\y$.  There are many approaches to solving this problem.  Typically we split this question into two parts: ({\em i}) What properties of $\bPhi$ will ensure that there exists {\em some} algorithm that can recover $x(t)$ from $\y$? and ({\em ii}) What algorithms can perform this recovery efficiently?

The answer to the first question is rather intuitive.  In order to simplify our discussion, we observe that since $\bPhi$ is linear, we can write
\begin{align*}
y_j = \iprod{\phi_j(t)}{x(t)} & = \iprod{\phi_j(t)}{\sum_{k=0}^{\ambient-1} \alpha_k \psi_k(t)} \\
  & = \sum_{k=0}^{\ambient-1} \alpha_k \iprod{\phi_j(t)}{\psi_k(t)}.
\end{align*}
If we let $\R$ denote the $\meas \times \ambient$ matrix with entries given by $r_{jk} = \iprod{\phi_j(t)}{\psi_k(t)}$, then (\ref{eq:phi1}) reduces to
\begin{equation} \label{eq:phi2}
\y = \R \balpha + \e.
\end{equation}
Thus, our problem reduces to that of designing the matrix $\R$.\footnote{Note that we can always go from $\R$ back to $\bPhi$ by setting $\bPhi = \R \bPsi^*$, where $\bPsi^*$ denotes the adjoint of $\bPsi$.}  Although many properties for $\R$ have been studied, the most common is the {\em restricted isometry property} (RIP)~\cite{CandeT_Decoding}.  The RIP of order $\sparsity$ requires that there exists a constant $\delta \in (0,1)$ such that for all $\sparsity$-sparse $\balpha$,

\begin{equation}
\label{RIP} \sqrt{1-\delta} \le \frac{\norm{\R \balpha}}{\norm{\balpha}} \le \sqrt{1+\delta}.
\end{equation}
In words, $\R$ preserves the Euclidean norm of $\sparsity$-sparse vectors.  Equivalently, the RIP of order $2\sparsity$ ensures that $\R$ preserves the Euclidean distance between pairs of $\sparsity$-sparse vectors. The RIP provides a guarantee that any $\sparsity$-sparse signal consistent with measurements perturbed as in~\eqref{eq:phi2} will be close to the original signal, and so the RIP ensures that the system has a degree of robustness to measurement noise.

We now consider how to design a matrix $\R$ satisfying the RIP.  An important result from the CS theory is that if the entries $r_{ij}$ are independent realizations from a Gaussian, Rademacher ($\pm 1$-valued), or more generally, any bounded, zero-mean distribution, then with overwhelmingly high probability $\R$ will satisfy the RIP of order $\sparsity$ provided that
\begin{equation} \label{eq:Qc}
\subsamp \le \subsampcs = \frac{\kappa_0}{\log \subsampmax} \subsampmax,
\end{equation}
where $\kappa_0 < 1$ is a constant that depends on $\ambient$ and the probability with which (\ref{RIP}) holds~\cite{BaranDDW_Simple}.   From this we conclude that CS-based measurement operators $\bPhi$ pay a small penalty in terms of $\rho$ (of $\kappa_0/\log(\subsampmax)$), for not exploiting any {\em a priori} knowledge of the locations of the nonzero frequencies.

In general, the theoretical analysis is somewhat lacking concerning the precise value of $\kappa_0$. If a specific value for $\kappa_0$ is required, then one must typically determine it experimentally.  This can be accomplished via Monte Carlo simulations that identify how many measurements are sufficient to ensure exact recovery in the noise-free setting on at least, say, 99\% of trials.  As an example, it is shown in~\cite{TroppLDRB_Beyond} that $\subsamp \lesssim 0.6 \subsampmax/\log(\subsampmax)$ is sufficient to enable exact recovery in the noise-free setting.  Note that this is not the same as demonstrating that $\subsamp \lesssim 0.6 \subsampmax/\log(\subsampmax)$ is sufficient to ensure that $\R$ satisfies the RIP, but it is suggestive that the true value of $\kappa_0$ is much greater than the conservative estimates provided by the theory. We will observe this phenomenon for ourselves in Section~\ref{sec:tests}.

Since the random matrix approach is somewhat impractical to build in hardware, several hardware architectures have been implemented and/or proposed that enable compressive samples to be acquired in practical settings. Examples include the random demodulator~\cite{TroppLDRB_Beyond}, random filtering~\cite{TroppWDBB_Random,LaskaSB_Polyphase}, the modulated wideband converter~\cite{MishaE_From}, random convolution~\cite{BajwaHRWN_Toeplitz,Rombe_Compressive}, the compressive multiplexer~\cite{SlaviLDB_Compressive}, and more~\cite{YuHS_Mixed-Signal,LexaDT_Reconciling}. Mathematically, such systems can typically be represented as matrices that operate on the vector $\x$ of Nyquist-rate samples. The corresponding matrices, although randomized, typically exhibit a great deal of structure. Although theoretical analysis of structured random matrices remains a topic of active study in the CS community, there do exist theoretical guarantees for some of these architectures~\cite{TroppLDRB_Beyond,BajwaHRWN_Toeplitz,Rombe_Compressive,SlaviLDB_Compressive}.  The amount of subsampling possible with these constructions is generally consistent with fully random measurements as given in (\ref{eq:Qc}), although proofs that these constructions satisfy the RIP typically result in raising the denominator in (\ref{eq:Qc}) to a small power (e.g., 2 or 4).

\subsection{CS recovery algorithms}
\label{subsec:rec}

We now address the question of how to recover the signal $x(t)$ from the measurements $\y$.  Most algorithms actually provide a recovery of $\balpha$, denoted $\widehat{\balpha}$.  By setting $\widehat{x}(t) = \bPsi(\widehat{\balpha})$, we can recover  $x(t)$. Alternatively, if we build a matrix $\bPsi$ by choosing the $k^\mathrm{th}$ column to be uniform samples of the function $\psi_{k}(t)$, then we can recover the Nyquist-rate samples of $x(t)$ directly by $\widehat{\x} = \bPsi \widehat{\balpha}$. Thus, $x(t)$, $\x$, and $\balpha$ are effectively equivalent representations of the signal of interest.

The original CS theory proposed $\ell_1$-minimization as a recovery technique when dealing with noise-free measurements~\cite{Cande_Compressive,Donoh_Compressed}.  Noisy measurements as in (\ref{eq:phi2}) can be easily handled using similar techniques provided that the noise $\e$ is bounded, meaning that $\norm{\e} \leq \epsilon$.  The accuracy of this method is made precise in~\cite{CandeRT_Stable}, which establishes that for $\sparsity$-sparse signals, provided $\R$ satisfies the RIP of order $2\sparsity$, the recovery error can be bounded by
\begin{equation} \label{eq:thm1}
\norm{\widehat{\balpha} - \balpha} \leq \kappa_1 \epsilon,
\end{equation}
where $\kappa_1 \ge 1$ is a constant that depends on the subsampling factor $\subsamp$.  The optimal value of this constant is similarly difficult to determine analytically, but in practice it should be close to 2 provided that $\subsamp < \subsampcs$.  Thus, measurement noise has a controlled impact on the amount of noise in the reconstruction.  A similar guarantee can be obtained for approximately sparse, or {\em compressible}, signals.\footnote{By compressible, we mean signals that are well approximated by a sparse signal.  The guarantee for compressible signals is similar to (\ref{eq:thm1}) but includes an additional term that quantifies the error incurred by approximating $x(t)$ with a sparse signal.}

$\ell_1$-minimization techniques are powerful methods for CS signal recovery, but there also exist a variety of alternative algorithms that are commonly used in practice and for which performance guarantees comparable to that of (\ref{eq:thm1}) can be established.  As an example, CoSaMP is a greedy algorithm known to satisfy guarantees similar to (\ref{eq:thm1})~\cite{NeedeT_CoSaMP}. CoSaMP and similar greedy algorithms build on a common set of simple techniques and can be easily understood by breaking the recovery problem into two separate sub-problems: identifying the locations of the nonzero coefficients of $\balpha$ and estimating the values of the nonzero coefficients of $\balpha$.  The former problem is clearly somewhat challenging, but once solved, the latter is relatively straightforward and can be solved using standard techniques like least squares.  In particular, suppose that an {\em oracle} provides us with the indices of the nonzero coefficients of $\balpha$ (its {\em support}), denoted by the index set $\Lambda$. Then a natural recovery strategy is to solve the problem:
\begin{equation}
\label{eq:oracle}
\widehat{\balpha} = \mathop{\argmin}_{\balpha} \ \norm{\R \balpha - \y } \quad \mathrm{s.t.} \quad \supp(\balpha) = \Lambda.
\end{equation}
If we let $\R_\Lambda$ denote the submatrix of $\R$ that contains only the columns of $\R$ corresponding to the index set $\Lambda$, then the solution to (\ref{eq:oracle}) is obtained via the pseudoinverse of $\R_\Lambda$, denoted $\R_\Lambda^\dagger$, i.e.,
\begin{equation}
\widehat{\balpha}|_{\Lambda} = \R_{\Lambda}^\dag \y \quad \quad \mathrm{and} \quad \quad \widehat{\balpha}|_{\Lambda^c} = 0.
\label{eq:ls2}
\end{equation}
Note that in the noise-free setting, if the oracle provides the correct $\Lambda$, then $\y = \R_\Lambda \balpha$, and so plugging this into (\ref{eq:ls2}) yields $\widehat{\balpha} = \balpha$ provided that $\R_\Lambda$ is full rank.  Thus, the central challenge in recovery is to correctly identify the support.  CoSaMP and related algorithms solve this problem by iteratively identifying likely nonzeros, estimating their values, and then improving the estimate of the support of the signal.

\section{Impact of White Noise on CS-Based Acquisition Systems}
\label{sec:whitenoise}

\subsection{Noise model and performance metrics}

The bulk of the CS literature focuses on CS acquisition and recovery in the face of {\em measurement noise} $\e$ as in (\ref{eq:phi2})~\cite{CandeRT_Stable,NeedeT_CoSaMP,BlumeD_Iterative,CohenDD_Instance,CandeT_Dantzig}.  In this section we build on~\cite{TreicDB_Application} to examine the effect of both measurement noise (i.e., noise caused by the measurement hardware) as well as {\em signal noise} (noise at the antenna induced by the fact that the receiver's physical temperature is above absolute zero combined with other random disturbances in the channel). Specifically, we now assume that the Fourier representation of $x(t)$ consists of a $\sparsity$-sparse signal corrupted by additive noise $\n$.  Thus, we acquire the measurements
$$
\y = \R (\balpha + \n) + \e = \R \balpha + \R \n + \e.
$$
This scenario is subtly different from (\ref{eq:phi2}), because part of the noise is now scaled by the matrix $\R$.  Our chief interests are to understand how $\R$ affects $\n$ and how $\R \n + \e$ manifests itself in the recovered $\widehat{\balpha}$.

It is commonly assumed that $\n = 0$ and that $\e$ is an arbitrary bounded vector.  In many common settings it is more natural to assume that $\e \sim \sN(0,\sigma_{\e}^2 \I_{\meas})$, i.e., $\e$ is i.i.d.\ Gaussian noise, and likewise for $\n$.  In this section we will consider the more general setting where $\e$ is a zero-mean, {\em white} random vector, meaning that
\begin{equation} \label{eq:whitenoisemean}
\expec{\e} = 0 \quad \quad \mathrm{and} \quad \quad \expec{ \e \e^T } = \sigma_{\e}^2 \I_{\meas}.
\end{equation}
Similarly, we will assume that $\n$ is also a zero-mean, white random vector (with $\expec{ \n \n^T } = \sigma_{\n}^2 \I_{\ambient}$).  Note that this means that the noise is added across the entire $\ambient$-dimensional Fourier spectrum and not just to the $\sparsity$ nonzeros of $\balpha$.

In order to quantify the impact of noise on our measurement process, we will define a variety of different {\em signal-to-noise ratios} (SNRs).  To begin, we recall a classical method for calculating the SNR of the original signal, which we will refer to as the {\em input SNR} (ISNR). The ISNR as defined here\footnote{Our definition corresponds to the so-called {\em in-band SNR}.  An alternative approach would be to consider the {\em out-of-band SNR}, which includes the entire noise vector across the full bandwidth. For white noise, these definitions are essentially equivalent but will differ by a factor of $\sparsity/\ambient$.} measures the SNR by including only the noise within the same bandwidth as the signal, which in our notation is given by
\begin{equation} \label{eq:IBSNR}
\mathrm{ISNR} = \frac{\norm{\balpha}^2}{\expec{\norm{(\balpha + \n)|_\Lambda - \balpha}^2}} = \frac{\norm{\balpha}^2}{\expec{\norm{\n|_\Lambda}^2}},
\end{equation}
where $\Lambda$ represents the support of $\balpha$.\footnote{Note that in this paper the signal $\balpha$ is deterministic, so the SNR will depend on our particular choice of $\balpha$.} The ISNR is a measure of the SNR of the signal itself prior to the acquisition of any measurements.

In the context of CS, it is also useful to consider two additional notions of SNR.  We define the {\em measurement SNR} (MSNR) as
\begin{equation} \label{eq:MSNR}
\mathrm{MSNR} = \frac{\norm{\R \balpha}^2}{ \expec{\norm{\y - \R \balpha}^2} } = \frac{\norm{\R \balpha}^2}{ \expec{\norm{\R \n + \e}^2} }
\end{equation}
and the {\em recovered SNR} (RSNR) as
\begin{equation} \label{eq:RSNR}
\mathrm{RSNR} = \frac{\norm{\balpha}^2}{ \expec{\norm{ \widehat{\balpha} - \balpha}^2} },
\end{equation}
where $\widehat{\balpha}$ is the output of our CS recovery algorithm.  In this section we will analyze the RSNR where $\widehat{\balpha}$ is the output of the oracle-assisted recovery algorithm in (\ref{eq:ls2}) applied to $\y = \R (\balpha + \n) + \e$.  We focus on the oracle-assisted recovery algorithm in order to simplify the analysis while illustrating a ``best-case'' scenario (the oracle-assisted approach yields the same output as a greedy method like CoSaMP when the greedy method correctly identifies the true support).  Note that in this setting, if we were to directly acquire the full signal, i.e., set $\R = \I_{\ambient}$, then (ignoring the impact of the measurement noise, which in this case would simply be combined with the signal noise term) the RSNR would be identical to the ISNR.  Since the oracle knows which elements should be zero, it is able to achieve zero error on those coefficients --- the only impact of the noise is on the nonzero coefficients.

\subsection{Impact of white measurement noise on the RSNR}

We begin by examining the impact on the RSNR of zero-mean, white measurement noise, i.e., the case where $\e$ satisfies (\ref{eq:whitenoisemean}). To isolate the impact of measurement noise, for the moment we will assume that $\n= \boldsymbol{0}$.  In light of results such as (\ref{eq:thm1}), one might expect that in general we will have $\norm{ \widehat{\balpha} - \balpha}^2 \approx \norm{\e}^2$. Since we assume that $\R$ satisfies the RIP, we will also have $\norm{\R \balpha}^2 \approx \norm{\balpha}^2$, which would lead to the conclusion that RSNR $\approx$ MSNR.  However, we will now see that we actually can expect that the RSNR will be increased by a factor of roughly $M/W$ compared to the MSNR which corresponds to the potential denoising of white noise that can occur when the signal is known to live in a $W$-dimensional subspace.  The proof of this and many of the more technical remaining theorems can be found in the Appendix.

\begin{thm} \label{thm:expectederror}
Suppose that $\y = \R \balpha + \e$, where $\e \in \real^M$ is a zero-mean, white random vector whose entries have variance $\sigma_{\e}^2$ and $\balpha$ is $\sparsity$-sparse.  Furthermore, suppose that $\R$ satisfies the RIP of order $\sparsity$ with constant $\delta$.  Then the $\widehat{\balpha}$ provided by the oracle-assisted recovery algorithm (\ref{eq:ls2}) satisfies
\begin{equation} \label{eq:thm:expectederror_a}
\frac{\sparsity \sigma_{\e}^2}{1+\delta} \le \expec{ \norm{\widehat{\balpha} - \balpha}^2 } \le \frac{\sparsity \sigma_{\e}^2}{1-\delta}.
\end{equation}
Hence, the RSNR of the oracle-assisted recovery algorithm satisfies
\begin{equation} \label{eq:thm:expectederror_b}
\left( \frac{1-\delta}{1+\delta} \right) \frac{M}{W}  \le \frac{\mathrm{RSNR}}{\mathrm{MSNR}} \le \left( \frac{1+\delta}{1-\delta} \right) \frac{M}{W}.
\end{equation}
\end{thm}

Again, by bounding the performance of the oracle-assisted recovery algorithm, we provide a best-case analysis of how an algorithm like CoSaMP would perform were it able to correctly identify the true support of $\balpha$.  Intuitively, Theorem~\ref{thm:expectederror} captures the notions that the RSNR should scale with the MSNR and that for a constant MSNR additional measurements should lead to an increasingly accurate recovery.

\subsection{Impact of white signal noise on the MSNR}

We now examine the scenario where the input signal itself is contaminated with noise.  Our chief interest is to understand how $\R$ impacts the signal noise. In order to simplify our analysis, we will make two assumptions concerning $\R$: ({\em i}) the rows of $\R$ are orthogonal and ({\em ii}) each row of $\R$ has equal norm.  Although these assumptions are not necessary to ensure that $\R$ satisfies the RIP, both are rather intuitive.  For example, it seems reasonable that if we wish to take the minimal number of measurements, then each measurement should provide as much new information about the signal as possible, and thus requiring the rows of $\R$ to be orthogonal seems natural.  The second assumption is similarly intuitive and can simply be interpreted as requiring that each measurement have ``equal weight''.  Note that if $\R$ is an orthogonal projection, then it automatically satisfies these properties.  These assumptions also hold for all of the $\R$ matrices corresponding to the practical architectures described in Section~\ref{ssec:csthy}.  Moreover, given an arbitrary matrix $\R$ that satisfies the RIP, it is always possible to construct a matrix $\widetilde{\R}$ that has the same row space as $\R$ and does satisfy these properties, as shown in the following lemma.

\begin{lemma} \label{lem:RIP_orthoprojector}
Suppose that  $\R \in \real^{\meas \times \ambient}$ satisfies the RIP of order $\sparsity$ with constant $\delta$.  There exists a matrix $\widetilde{\R}$ with orthonormal rows and the same row space as $\R$ such that for all $\sparsity$-sparse $\balpha$
\begin{equation} \label{eq:RIP_ortho_bound}
\frac{\sqrt{1-\delta}}{s_{\mathrm{max}}(\R)}  \le \frac{\norm{\widetilde{\R} \balpha }}{\norm{\balpha}} \le \frac{\sqrt{1+\delta}}{s_{\mathrm{min}}(\R)},
\end{equation}
where $s_{\mathrm{max}}(\R)$ denotes the largest and $s_{\mathrm{min}}(\R)$ the smallest nonzero singular values of $\R$.
\end{lemma}

Lemma~\ref{lem:RIP_orthoprojector} constructs matrices $\widetilde{\R}$ with unit-norm rows.  Note that if $\R$ has unit-norm columns and also has orthogonal rows of equal norm, then the rows must have norm $\sqrt{\subsamp}$.  In this case, $s_{\mathrm{max}}(\R) = s_{\mathrm{min}}(\R) = \sqrt{1/\subsamp}$, and thus (\ref{eq:RIP_ortho_bound}) is simply a rescaled version of the RIP (and $\widetilde{\R}$ is a rescaled version of $\R$).  In general we will have $s_{\mathrm{max}}(\R) > s_{\mathrm{min}}(\R)$, in which case Lemma~\ref{lem:RIP_orthoprojector} indicates that the constants in the bound could become slightly worse.  Importantly, however, for randomly generated $\R$ matrices it can be shown that, provided $\meas \ll \ambient$, $s_{\mathrm{max}}(\R) \approx s_{\mathrm{min}}(\R) \approx \sqrt{1/\subsamp}$ with high probability, and so for these $\R$ we should not expect significant degradation.  Thus, without loss of generality we now restrict our attention to matrices $\R$ which have orthogonal rows each of norm $\sqrt{\subsamp}$.  This assumption ensures that if $\n$ is white noise, then $\R \n$ will also be white, allowing us to easily analyze the impact of signal noise in the following theorem.  To isolate the impact of $\n$ we assume that $\e = \boldsymbol{0}$.

\begin{thm} \label{thm:noiseexpansion}
Suppose that $\y = \R( \balpha + \n)$, where $\n \in \real^B$ is a zero-mean, white random vector whose entries have variance $\sigma_{\n}^2$ and $\balpha$ is $\sparsity$-sparse.  Furthermore, suppose that $\R$ satisfies the RIP of order $\sparsity$ with constant $\delta$ and has orthogonal rows, each of norm $\sqrt{\subsamp}$. Then $\R \n$ is also a zero-mean, white random vector whose entries have variance $\sigma_{\R \n}^2 = \subsamp \sigma_{\n}^2$, and hence
\begin{equation} \label{eq:noiseexpansionb}
(1-\delta) \frac{\sparsity}{\ambient} \le \frac{\mathrm{MSNR}}{\mathrm{ISNR}} \le (1+\delta) \frac{\sparsity}{\ambient}.
\end{equation}
\end{thm}

Thus, although the oracle-assisted recovery procedure served to mildly attenuate white noise added to the measurements, noise added to the signal itself can be highly amplified by the measurement process when $\meas \ll \ambient$ (or $\subsamp \gg 1$).  This is directly analogous to a classical phenomenon known as {\em noise folding}.

\subsection{Noise folding in CS}

Theorem~\ref{thm:noiseexpansion} tells us that the $\R$ matrices used in CS will amplify white noise by a factor of $\subsamp$. We can quantify the impact of noise folding by considering the ratio of the ISNR to the RSNR.  This ratio quantifies the penalty for using CS when the support of $\balpha$ is actually known {\em a priori}.

\begin{thm} \label{thm:noisefolding}
Suppose that $\y = \R( \balpha + \n)$, where $\n \in \real^B$ is a zero-mean, white random vector and $\balpha$ is $\sparsity$-sparse.  Furthermore, suppose that $\R$ satisfies the RIP of order $\sparsity$ with constant $\delta$ and has orthogonal rows, each of norm $\sqrt{\subsamp}$. Then the RSNR of the oracle-assisted recovery algorithm (\ref{eq:ls2}) satisfies
\begin{equation}
\frac{\subsamp}{1+\delta} \le \frac{\mathrm{ISNR}}{\mathrm{RSNR}} \le \frac{\subsamp}{1-\delta}.
\end{equation}
\end{thm}
\begin{proof}
We begin by observing that
$$
\frac{\mathrm{ISNR}}{\mathrm{RSNR}} = \frac{\expec{\norm{\widehat{\balpha} - \balpha}^2}}{\expec{\norm{\n|_\Lambda}^2}}.
$$
Since $\n$ is white, we have that
\begin{equation} \label{eq:expectedsnrloss1}
\expec{\norm{\n|_\Lambda}^2} = \sparsity \sigma_{\n}^2.
\end{equation}

From Theorem~\ref{thm:noiseexpansion}, we observe that $\y = \R \balpha + \R \n$, where $\R \n$ is a white random vector with $\expec{\R \n (\R \n)^T} = \subsamp \sigma_{\n}^2 \I_{\meas}$.  Since $\R \n$ is white, we apply Theorem~\ref{thm:expectederror} to obtain
\begin{equation} \label{eq:expectedsnrloss2}
\frac{\sparsity \subsamp \sigma_{\n}^2}{1+\delta} \le \expec{\norm{\widehat{\balpha} - \balpha}^2} \le \frac{\sparsity \subsamp \sigma_{\n}^2}{1-\delta}.
\end{equation}
Taking the ratio of (\ref{eq:expectedsnrloss2}) and (\ref{eq:expectedsnrloss1}) and simplifying establishes the theorem.
\end{proof}
Note that this theorem could also be established (with a slightly worse constant) by simply combining~\eqref{eq:thm:expectederror_b} and~\eqref{eq:noiseexpansionb}.

Noise folding has a significant impact on the amount of noise present in CS measurements. Specifically, if we measure the ratio in dB, then we have that
$$
\frac{\mathrm{ISNR}}{\mathrm{RSNR}} \approx 10 \log_{10} \left(\subsamp \right).
$$
Thus, every time we double the subsampling factor $\subsamp$ (a one octave increase), the SNR loss increases by $3$dB. In other words, {\em for the acquisition of a sparse signal in white noise, the RSNR of the recovered signal decreases by 3dB for every octave increase in the amount of subsampling}.

The impact of noise folding on the recovery error in CS systems satisfying the RIP was first reported and analyzed in~\cite{TreicDB_Application}, while further analysis was conducted in~\cite{Daven_Random}.  A closely related issue was also considered in~\cite{AeroSZ_Information}, which studies the impact of signal noise on the problem of exactly recovering the true support of the original signal and reaches similar conclusions.  Specifically,~\cite{AeroSZ_Information} shows that exact support recovery in the presence of signal noise requires that $\meas$ scales linearly with $\ambient$.  Our main result says that in order to achieve a desired RSNR, we will also require that $\meas$ scales linearly with $\ambient$.

During the final preparation of this manuscript similar results were established for systems with bounded {\em coherence}~\cite{CastroYoninaNoiseFold}.  We also note that alternative signal acquisition techniques like {\em bandpass sampling}\footnote{In practice bandpass sampling is not suitable for the typical CS settings.  This is because if there are multiple narrowband signals present in the bandwidth occuring at unknown frequencies, then bandpass sampling causes irreversible aliasing so that the components can potentially overlap and will be impossible to separate.  In contrast to bandpass sampling, CS acquisition preserves sufficient information to enable the recovery of both the values and the locations of the large Fourier coefficients.} (sampling a narrowband signal uniformly at a sub-Nyquist rate to preserve the values but not the locations of its large Fourier coefficients) are affected by an identical $3$dB/octave SNR degradation~\cite{VaughSW_Theory}.  In fact, recent results on minimax rates of estimation in high dimensions imply that, for Gaussian noise, any acquisition system that acquires fewer than $\ambient$ linear measurements will suffer from the same $3$dB/octave SNR loss, regardless of how the measurements are chosen and no matter how sophisticated the recovery algorithm~\cite{RaskuWY_Minimax,CandeD_How}. More recently it has been shown that {\em even if the measurements are chosen adaptively}, this loss cannot be avoided~\cite{CastrCD_Fundamental}.

Noise folding is not restricted to only sparse signals in noise~\cite{DaviesGuo_sample}.  In fact, so-called \emph{compressible signals} are also subject to something akin to noise folding, even when the input contains no actual noise. Specifically, one can often model the ``tail'' of such signals as uncorrelated noise, enabling the application of the results above.  Treating the tail as noise, we would na\"{i}vely arrive at the same $3$dB SNR loss as above.  In fact, the loss resulting from ``tail-folding'' is even worse, since as we increase the amount of subsampling, we will be able to recover less of the signal (the range of $\sparsity$ we can handle will decrease), and thus the ``tail'' will encompass an ever larger portion of the signal.

The $3$dB/octave SNR degradation represents an important tradeoff in the design of CS receivers. This yields the the engineering design rule for CS receivers of $N\!F \approx 10 \log_{10}(\subsamp)$, where $N\!F$ is the noise figure as defined in Section~\ref{sec:reqs}. This result implies that for a fixed signal bandwidth $\sparsity/2$ there is a practical limit to the instantaneous bandwidth $\ambient/2$ for which we can obtain a desired RSNR.  In Section~\ref{sec:tests} we match this theoretical result against the results of multiple simulations.

Although the noise folding behavior of compressive systems imposes a very real cost, this does not necessarily preclude its use in practical systems, one example of which is discussed in Section~\ref{sec:conc}.  Conversely, the dramatic sampling rate reduction enabled by CS can lead, in some cases, to significant improvements in the dynamic range of the system.  This issue is examined in the next section.

\section{Dynamic Range of CS-Based Acquisition Systems}
\label{sec:dr}

A fundamental advantage of CS is that it enables a significantly lower sampling rate for sparse signals, which in turn enables the use of higher-resolution ADCs~\cite{LeRonRee::2005::Analog-to-Digital-Converters}.  By exploiting this fact, a CS acquisition system should exhibit a significantly larger dynamic range than a conventional system.  In this section, we provide a theoretical justification for this claim, and in Section~\ref{sec:tests} we will quantify the potential gain via the empirical relationship between sampling rate and quantizer resolution~\cite{LeRonRee::2005::Analog-to-Digital-Converters}.

We begin our analysis by first providing a rigorous and general definition of dynamic range.  Roughly, we define the dynamic range as the ratio of the maximum to the minimum signal power levels that can be handled with ``full fidelity''.\footnote{In this section we are analyzing the CS-based receiver's dynamic range {\em as a system}.  This should not be confused with the dynamic range {\em of a signal}, which in our framework could be quantified as the ratio of the largest to smallest entry in $\balpha$.  An examination of how CS impacts the maximum allowable {\em signal} dynamic range would be worthwhile, but it is not our focus in this paper.  We do note, however, that some results in the present section could play a key role in any such analysis.}  In order to make this notion precise, we will ignore the effects of any noise or nonlinearities from the other ADC components and examine only the impact of the quantizer.  This is a fair assumption, since a key goal in the design of an ADC is that the quantizer be the only component that limits the device's dynamic range.

Our definition of dynamic range has two properties that aid us in the analysis of CS systems: ({\em i}) the dynamic range does not depend on a stochastic quantization error model, and ({\em ii}) any reduction of quantization error yields a corresponding improvement in dynamic range, i.e., the dynamic range of the quantizer effectively determines the dynamic range of the system.  With this definition in hand, we examine quantization in both conventional and CS systems and provide lower bounds on the dynamic range of each.  Our key finding will be that, all things being equal, the dynamic range of a CS acquisition system is generally no worse than that of a conventional system.  Thus, since CS enables lower sampling rates for sparse signals, we can employ a higher-resolution ADC and attain a larger dynamic range.

\subsection{A deterministic approach to dynamic range}
\label{subsec:dr}

\begin{figure}[tbp]
   \centering
   \includegraphics[width=2in]{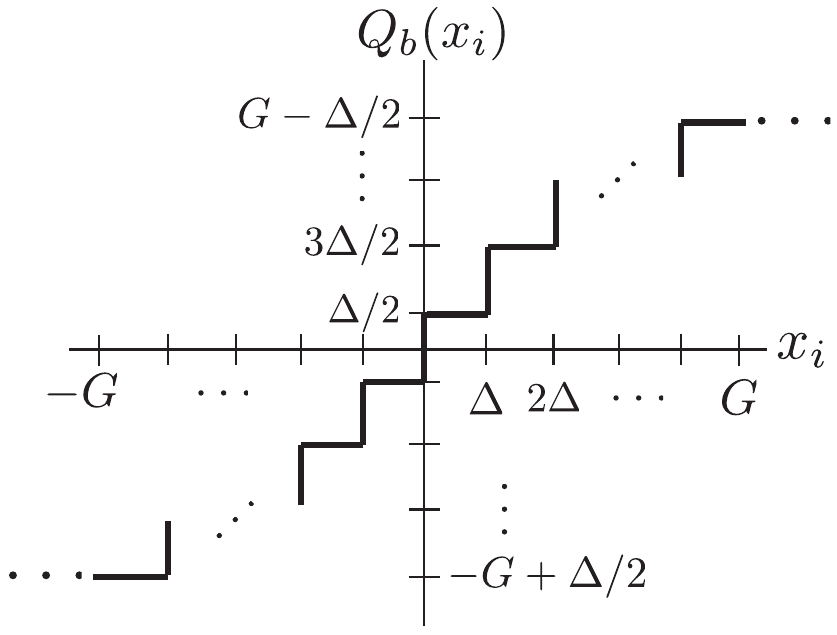}
   \caption{A midrise uniform quantization function $Q_{b}(x_{i})$ with $b$ bits, saturation level $G$, and quantization interval $\Delta = 2^{-b+1}G$.}
   \label{fig:midrisequant}
\end{figure}

To formulate our definition of dynamic range, we first analyze the error induced by quantizing $\x$.
We assume that $\x$ is a vector in $\real^\ambient$ and compare $\x$ with the $b$-bit quantized version of $\x$, which we denote by $\quant{\x}$.  Let $\Delta$ denote the quantization interval, and let $\pm G$ denote the saturation levels, so that $G = \Delta 2^{b-1}$.  Note that if $|x_i| \le G$, then we have that $| x_i - \quant{x_i}| \le \Delta/2$, but if $|x_i| > G$ then $|x_i - \quant{x_i}| = |x_i|- (G-\Delta/2)$.  A midrise uniform quantization function $Q_{b}(x_{i})$ is depicted in Fig.~\ref{fig:midrisequant}.

For a given $\x$, we define the {\em signal-to-quantization noise ratio} (SQNR) of the quantizer as
\begin{equation} \label{eq:SQNR}
\mathrm{SQNR}(\x) = \frac{\norm{\x}^2}{\norm{\x - \quant{\x}}^2}.
\end{equation}
We make the dependence of the SQNR on $\x$ explicit, since our definition of dynamic range will be based on how the scaling of $\x$ affects the SQNR.  First, however, we establish a practical bound on the best SQNR attainable for a given $G$, $\Delta$, and $\x$.
\begin{lemma} \label{lem:betaopt}
Let $\x \in \real^\ambient$ be arbitrary.  There always exists a $\beta > 0$ such that
\begin{equation} \label{eq:SQNRformula}
\mathrm{SQNR}(\beta \x) \ge \frac{1}{\gamma(\x)^2} \left(  \frac{2G}{\Delta} \right)^2,
\end{equation}
where
\begin{equation} \label{eq:gammadef}
\gamma(\x) = \frac{\norm[\infty]{\x}}{\norm{\x}/\sqrt{\ambient}}.
\end{equation}
\end{lemma}
The quantity $\gamma$ in (\ref{eq:gammadef}) is known as the {\em peak-to-average ratio} (PAR) of $\x$. Also known as the {\em crest factor} or {\em loading factor}~\cite{Lyons_Understanding}, it is a measure of the ratio between a signal's ``average energy'' to its peak value.

While the expression in (\ref{eq:SQNRformula}) may look foreign to some, this bound is similar to standard results for peak SQNR.  Recall that $2G/\Delta = 2^b$.  Thus, if we express (\ref{eq:SQNRformula}) in dB, then we observe that by setting $\beta$ appropriately we can obtain
\begin{align}
\mathrm{SQNR}(\beta \x) & \geq  20 b \log_{10}(2) - 20 \log_{10} (\gamma(\x)) \notag \\
& \gtrsim 6.02 b - 20 \log_{10} (\gamma(\x)). \label{eq:SQNRdb}
\end{align}
This corresponds to the well-known result that the peak SQNR grows by approximately $6$dB per quantizer bit~\cite{Lyons_Understanding}.  Furthermore, although the SQNR bound in (\ref{eq:SQNRdb}) provides only a {\em lower bound} on the SQNR, it generally agrees with the results in the literature that assume probabilistic models on the signal $\x$ and/or the quantization noise.  For example, a more conventional probabilistic analysis would assume that the quantization noise has a uniform distribution.  In this case, one can derive the expression
$$
\expec{\mathrm{SQNR}(\beta \x)} \approx 6.02 b - 20 \log_{10}(\gamma(\x)) + 4.77,
$$
where the additive constant $4.77$  reflects the improvement made possible over our worst-case bound by placing a uniform distribution on the quantization noise~\cite{Lyons_Understanding}.  For our purposes below, a lower bound on the SQNR is sufficient.  We view the deterministic nature of our bound as a strength, allowing us to avoid any questionable assumptions concerning the quantization noise distribution.

We now show how we can use the SQNR to offer a concrete definition for dynamic range.  Specifically, suppose that we would like to achieve an SQNR of at least $C$.\footnote{In our analysis we consider $C \in \left( 1, (2G/\Delta)^2/\gamma(\x)^2 \right]$ to ensure that our definition leads to a meaningful notion of dynamic range.  Specifically, once we fix $\Delta$ and $G$, there is an upper limit on the SQNR we can hope to achieve, and for $C$ beyond that limit the dynamic range will be ill-defined.  Similarly, if we set $C=1$ then one can easily achieve infinite dynamic range by quantizing all signals to zero. However, for the range of $C$ considered we can always set  $\beta = G/\norm[\infty]{\x}$ to ensure that $\mathrm{SQNR}(\beta \x) \ge C$ (see the proof of Lemma~\ref{lem:betaopt}).}  We aim to identify the range of scalings $\beta$ of a given signal $\x$ for which $\mathrm{SQNR}(\beta \x) \ge C$. More formally, we want to find scalars $\betamin(\x)$ and $\betamax(\x)$ such that $\mathrm{SQNR}(\beta \x) \ge C$ for all $\beta \in \left[\betamin(\x),\betamax(\x) \right]$, where
\begin{equation}
\betamin(\x) \le G/\norm[\infty]{\x} \le \betamax(\x).
\end{equation}
In words, $\betamax(\x)$ and $\betamin(\x)$ define a range of scalings over which we achieve the desired SQNR $C$.

Using $\betamax(\x)$ and $\betamin(\x)$, we define the dynamic range of a conventional acquisition system as
\begin{equation} \label{eq:DRDEF}
\mathrm{DR}_C(\x) := \left( \frac{\betamax(\x)}{\betamin(\x)} \right)^2.
\end{equation}
Hence, the dynamic range of a conventional ADC is the ratio of the maximum input scaling to the minimum input scaling of $\x$ such that for both scalings the SQNR is at least $C$. Note that this definition implicitly assumes that the entire vector $\x$ is scaled by a single constant that does not dynamically change with time, which might seem to disallow the possibility of automatic gain control (AGC).  However, in practice the time constant of any pre-quantizer AGC circuit will likely be many orders of magnitude larger than the block-length, so the use of an AGC is largely irrelevant to our analysis.

At first sight, (\ref{eq:DRDEF}) may appear to be a rather complicated way of describing what is at heart an elementary concept --- dynamic range is often simply quantified as the ratio of the largest to smallest quantization levels.  However, the strength of this definition is that it can easily be extended to quantify the dynamic range of a CS-based ADC in which the measurement and recovery processes obscure the impact of finite-range quantization on the final RSNR. Specifically, given an input signal $\x$ (or equivalently $\balpha$) we apply a recovery algorithm to the quantized CS measurements $\quant{\y} = \quant{\R \balpha}$ to obtain a recovery $\widehat{\balpha}$.  We wish to understand the impact of this quantization on the resulting RSNR.  While it might not otherwise be immediately apparent, (\ref{eq:DRDEF}) suggests a natural way to extend the definition of dynamic range to the CS setting by simply replacing RSNR with SQNR, i.e., defining $\betamin(\x)$ and $\betamax(\x)$ by considering the range of scalars $\beta$ such that $\mathrm{RSNR}(\beta \x) \ge C$.  Note that for a conventional ADC, since $\mathrm{RSNR}=\mathrm{SQNR}$, the definition remains unchanged from (\ref{eq:DRDEF}).  We now analyze the dynamic range of a conventional acquisition system in Section~\ref{ssec:convDR} and then extend this to the CS setting in Section~\ref{ssec:CSDR}.

\subsection{Dynamic range of a conventional ADC}
\label{ssec:convDR}

We first provide a simple bound on the dynamic range $\mathrm{DR}_C(\x)$ for a conventional ADC.
\begin{thm}
\label{thm:dr}
The dynamic range of a quantizer as defined by (\ref{eq:DRDEF}) is bounded by
\begin{equation} \label{eq:DR1}
\mathrm{DR}_{C}(\x) \ge \frac{1}{C \gamma(\x)^2 - 1} \left( \left(\frac{2G}{\Delta}\right)^{2} -1 \right),
\end{equation}
where $\gamma(\x)$ is defined as in (\ref{eq:gammadef}).
\end{thm}

For large $b$, $(2G/\Delta)^2 - 1 \approx 2^{2b}$, and so by expressing (\ref{eq:DR1}) in dB we obtain
\begin{equation} \label{eq:DR1db}
\mathrm{DR}_C(\x) \gtrsim 6.02b - 10\log_{10} \left( C \gamma(\x)^2 - 1 \right).
\end{equation}
This coincides with the familiar rule of thumb that just like the SQNR in (\ref{eq:SQNRdb}), ADC dynamic range increases by $6$dB per quantizer bit~\cite{Lyons_Understanding}. Note, however, that we again have an additive constant that here depends both on the targeted SQNR $C$ as well as the PAR $\gamma(\x)$.  This is again expected, since a more ambitious SQNR is more difficult to achieve and since a signal with higher PAR is harder to quantize, both of which lead to a more limited dynamic range.  We revisit the issue of PAR below in Section~\ref{ssec:PAR}.

In summary, our definition of dynamic range (\ref{eq:DRDEF}) yields a reasonable expression (\ref{eq:DR1}) for a conventional ADC that coincides with the traditional ``folk wisdom'' on dynamic range.

\subsection{Dynamic range of a CS-based acquisition system}
\label{ssec:CSDR}

Thus far we have proposed a rigorous and general definition of dynamic range and analyzed a conventional ADC in this context.  We now extend this analysis to the CS case.  Our argument proceeds by first showing that we can always relate $\mathrm{RSNR}(\beta \x )$ to $\mathrm{SQNR}(\beta \y)$ and then relate $\mathrm{SQNR}(\beta \y)$ to $\mathrm{SQNR}(\beta \x)$.  This allows us to argue that, whenever $\mathrm{SQNR}(\beta \x) > C$, we have that $\mathrm{RSNR}(\beta \x) > C'$ for some $C'$.  In other words, whenever we can achieve a certain SQNR $C$ by directly quantizing $\x$, a CS-based system can also achieve the RSNR $C'$ (where $C'$ is typically comparable to $C$).  Thus, the dynamic range of these systems will be essentially the same. We begin by relating $\mathrm{RSNR}(\beta \x )$ to $\mathrm{SQNR}(\beta \y)$.
\begin{lemma} \label{lem:QSNRCS}
Suppose that $\y = \R \balpha$, where $\balpha$ is $\sparsity$-sparse and $\R$ satisfies the RIP of order $\sparsity$ with constant $\delta$.  Let $\widehat{\balpha}$ denote the output of applying a recovery algorithm to the quantized measurements $\quant{\y}$ which satisfies a recovery guarantee like that given in (\ref{eq:thm1}), i.e.,
\begin{equation} \label{eq:QSNRCS1}
\norm{\widehat{\balpha} - \balpha}^2 \le \kappa_1^2 \norm{\quant{\y} - \y}^2.
\end{equation}
Then,
\begin{equation} \label{eq:QSNRCS2}
\mathrm{RSNR}(\beta \x) = \mathrm{RSNR}(\beta \balpha) \ge \frac{\mathrm{SQNR}(\beta \y)}{(1+\delta)\kappa_1^2}.
\end{equation}
\end{lemma}
\begin{proof}
Without loss of generality, suppose that $\beta = 1$.  From the RIP we have that
$$
\norm{\balpha}^2 \ge \frac{\norm{\R \balpha}^2}{1+\delta}.
$$
Combining this with (\ref{eq:QSNRCS1}), we obtain the bound
\begin{align*}
\mathrm{RSNR}(\x) = \mathrm{RSNR}(\balpha) & = \frac{\norm{\balpha}^2}{\norm{\widehat{\balpha} - \balpha}^2} \\
& \ge \frac{\norm{\y}^2}{(1+\delta)\kappa_1^2 \norm{\quant{ \y} - \y}^2} \\
& = \frac{\mathrm{SQNR}(\y)}{(1+\delta)\kappa_1^2},
\end{align*}
which completes the proof.
\end{proof}

In words, $\mathrm{RSNR}(\beta \x )$ is lower bounded by a constant multiple of $\mathrm{SQNR}(\beta \y)$.  This means that we can expect the RSNR to follow the same trend as the SQNR of the measurements.  Thus, we can restrict our analysis and comparisons to the measurement SQNR. Hence, we can compare $\mathrm{SQNR}( \beta \y)$ to $\mathrm{SQNR}(\beta \x)$ from Section~\ref{subsec:dr}.  The following lemma shows that we can bound $\mathrm{SQNR}( \beta \y)$ in a manner similar to how Lemma~\ref{lem:betaopt} bounds $\mathrm{SQNR}(\beta \x)$.
\begin{lemma}
Suppose that $\y = \R \balpha = \bPhi \x$, where $\balpha$ is $\sparsity$-sparse and $\R$ satisfies the RIP of order $\sparsity$ with constant $\delta$.  Then there exists a $\beta$ such that
$$
\mathrm{SQNR}(\beta \y) \ge (1-\delta) \rho \frac{\norm[\infty]{\x}^2}{\norm[\infty]{\y}^2} \frac{1}{\gamma(\x)^2} \left(  \frac{2G}{\Delta} \right)^2.
$$
\end{lemma}
\begin{proof}
We begin by noting that from Lemma~\ref{lem:betaopt}, $\beta = G/\norm[\infty]{\y}$ we have that
$$
\mathrm{SQNR}(\beta \y) \ge \left( \frac{\norm{\y}^2/\meas}{\norm[\infty]{\y}^{2}} \right) \left( \frac{2G}{\Delta} \right)^2.
$$
Since $\R$ satisfies the RIP we obtain
$$
\norm{\y}^2 \ge (1-\delta) \norm{\balpha}^2 = (1-\delta) \norm{\x}^2.
$$
Thus we have that
\begin{align*}
\frac{\norm{\y}^2/\meas}{\norm[\infty]{\y}^{2}} & \ge \frac{(1-\delta) \norm{\x}^2}{\norm[\infty]{\y}^2} \\
& =  (1-\delta) \left( \frac{\ambient}{\meas} \right) \frac{\norm[\infty]{\x}^2}{\norm[\infty]{\y}^2}  \left(\frac{\norm{\x}^2/\ambient}{\norm[\infty]{\x}^2} \right) \\
& = (1-\delta)\rho \frac{\norm[\infty]{\x}^2}{\norm[\infty]{\y}^2} \frac{1}{\gamma(\x)^2},
\end{align*}
which establishes the lemma.
\end{proof}

Thus, CS has the same $6$dB per quantizer bit behavior as in (\ref{eq:SQNRdb}) with
\begin{equation} \label{eq:CSSQNRdb}
\mathrm{SQNR}(\beta \y) \gtrsim 6.02 b + 20 \log_{10} \left( \frac{\sqrt{(1-\delta)\rho} \norm[\infty]{\x}}{\gamma(\x) \norm[\infty]{\y}} \right),
\end{equation}
the only difference being an additional additive constant that we will analyze in more detail in Section~\ref{ssec:PAR}.

We are now ready to compute the dynamic range of the CS acquisition system. We retain the same definition of dynamic range as in (\ref{eq:DRDEF}), but with $\betamax(\x)$ and $\betamin(\x)$ defined by substituting the SQNR constraint with the requirement that $\mathrm{RSNR}(\beta \x) \ge C$.  In this setting, we can repeat the same analysis as in Theorem~\ref{thm:dr} to obtain
$$
\mathrm{DR}_{C}(\x) \ge \frac{1}{C'\gamma(\x)^2 - 1} \left( \left(\frac{2G}{\Delta}\right)^{2} -1 \right),
$$
where
$$
C' = \frac{1-\delta}{(1+\delta)\kappa_1^2} \rho \frac{\norm[\infty]{\x}^2}{\norm[\infty]{\y}^2}.
$$
Thus, when measured in dB the dynamic range is affected by CS only through an additive constant.

In practice, we can take significant advantage of the fact that, all things being equal, a CS system has the same dynamic range as a conventional Nyquist ADC.  Specifically, because the ADC employed in a CS-based system operates at a significantly lower rate than would be required in a conventional system, a slower quantizer with higher bit-depth can be employed~\cite{LeRonRee::2005::Analog-to-Digital-Converters}.   If the gain in effective bits is large, then the $6$dB per bit improvement in dynamic range will dominate the additive constant and result in a substantial {\em increase} in the CS system's dynamic range as compared to a conventional ADC.  Moreover, this entire discussion has ignored the fact that CS systems are highly robust to large saturation errors due to the {\em democratic} nature of CS measurements~\cite{LaskaBDB_Democracy}, so these results may even be understating the possibility for improvement. We explore the possibility of increased dynamic range empirically in Section~\ref{sec:tests}.

\subsection{Impact of CS on the PAR}
\label{ssec:PAR}

We conclude this section with one last note regarding the mollifying effect of a CS acquistion system on the PAR.  All of our expressions for the SQNR or RSNR as well as the dynamic range of a system depend in some way on the PAR of the signal $\x$ or the measurements $\y$, depending on the context.  In practice, the PAR has a significant impact on the resulting expressions.  However, the PAR of a signal $\x$ can vary widely in the range
\begin{equation} \label{eq:gammabound}
1 \le \gamma(\x) \le \sqrt{B},
\end{equation}
which follows from standard norm inequalities.  As an example, combining (\ref{eq:gammabound}) with the lower bound on the SQNR of a conventional ADC in (\ref{eq:SQNRdb}) means that in the best case (which corresponds to an all-constant vector $\x$) the bound in (\ref{eq:SQNRdb}) reduces to $6$dB per bit growth in SQNR with no offset, whereas in the worst case (which corresponds to a $K=1$ sparse $\x$) we incur an additive penalty of $-10 \log_{10}(\ambient)$dB.  As the dimension $B$ grows this penalty can become large, reflecting the fact that as the number of samples grows it becomes possible to construct a signal that has ever larger PAR.  This translates to a similarly wide range of possible values for the additive penalty in the bound on dynamic range in (\ref{eq:DR1db}).

Our aim here is to understand how CS impacts PAR. Clearly, we expect the PAR of the CS measurements $\y$ to differ from that of the signal $\x$ since each measurement typically consists of a weighted sum of the entries of $\x$.  Intuitively, such measurements have the potential to average out some of the ``spikes'' in $\x$ resulting in a potentially improved PAR.  This appears in the analysis in the expression for $\mathrm{SQNR}( \beta \y)$ in (\ref{eq:CSSQNRdb}), which shows that $\mathrm{SQNR}( \beta \y)$ can be improved over $\mathrm{SQNR}(\beta \x)$ in (\ref{eq:SQNRdb}) if $\rho \norm[\infty]{\x}^2/\norm[\infty]{\y}^2$ is somewhat larger than 1. In the worst-case, this quantity can be a great deal smaller than 1; however, on average we are likely to do significantly better.  As an illustration, we  describe what can be said when $\bPhi$ is a matrix with i.i.d.\ $\pm 1/\sqrt{M}$ (Rademacher) entries, although our analysis could easily be adapted to show similar results for the practical architectures described in Section~\ref{ssec:csthy}.

We begin with the worst-case.  By combining the the Cauchy-Schwartz inequality with standard $\ell_p$-norm inequalities, we have that for all $j$, $|y_j| \le \rho \norm[\infty]{\x}.$  Thus we obtain
$$
\rho \frac{\norm[\infty]{\x}^2}{\norm[\infty]{\y}^2} \ge \frac{1}{\rho}.
$$
Hence, in the worst-case
$$
20 \log_{10} \left( \frac{\sqrt{(1-\delta)\rho} \norm[\infty]{\x}}{\norm[\infty]{\y}} \right) \approx - 10 \log_{10}(\rho),
$$
which corresponds to an SQNR loss of $3$dB per octave increase in the subsampling factor. However, this bound will be achieved only when $\x$ is both constant magnitude and has elements with signs exactly matching one of the (randomly chosen) rows of $\bPhi$ --- a highly unlikely scenario.  Furthermore, this bound makes no use of the ``dithering'' effect promoted by the randomized measurements, a grave omission indeed.  Towards this end, we next consider a probabilistic bound to see that we can typically obtain better performance.
\begin{lemma} \label{lem:probbound}
Suppose that $\bPhi$ is chosen with i.i.d.\ entries with variance $1/M$ drawn according to any strictly sub-Gaussian distribution.  Then
\begin{equation} \label{eq:inftybound}
\rho \frac{\norm[\infty]{\x}^2}{\norm[\infty]{\y}^2} \ge \frac{\gamma(\x)^2}{4 \log(M)}
\end{equation}
with probability at least $1 - 2/M$.
\end{lemma}
\begin{proof}
By combining the union bound (over $M$ measurements) with standard tail bounds on a strictly sub-Gaussian distribution, we obtain
$$
\prob{ \|\y\|_\infty > t} \le 2 M \exp \left( - \frac{M t^2}{2 \norm{\x}^2} \right).
$$
Thus,
\begin{align*}
\prob{\mathrm{(\ref{eq:inftybound})~does~not~hold} } & \le 2 M \exp \left( - \frac{4 M \rho \log(M) \norm[\infty]{\x}^2}{2 \gamma(\x)^2 \norm{\x}^2} \right) \\
 & = 2 M \exp \left( - \frac{2 \log(M)  \norm[\infty]{\x}^2}{\gamma(\x)^2 \norm{\x}^2/B} \right)   \\
& = 2 \exp \left(\log(M) - 2 \log(M) \right) = \frac{2}{M},
\end{align*}
which establishes the lemma.
\end{proof}

Thus, in practice we expect our bound for $\mathrm{SQNR}(\y)$ in (\ref{eq:CSSQNRdb})  to differ from our bound for $\mathrm{SQNR}(\x)$  in (\ref{eq:SQNRdb}) only by a factor of $\gamma(\x)^2/4 \log(M)$. Recalling our bound on $\gamma(\x)$ we have that
$$
\frac{1}{4 \log(M)} \le \frac{\gamma(\x)^2}{4 \log(M)} \le \frac{B}{4 \log(M)}.
$$
Hence, for $\x$ with small PAR, we can expect a potential loss in SQNR when compared to direct quantization of $\x$, while for $\x$ with moderate or large PAR we can actually expect a significant improvement.

Finally, we can use Lemma~\ref{lem:probbound} to approximate (\ref{eq:CSSQNRdb}) with high probability as
$$
\mathrm{SQNR}(\beta \y) \gtrsim 6.02 b - 20 \log_{10}\left( 4 \log(M)/\sqrt{1-\delta} \right),
$$
which implies that CS allows us to essentially eliminate the negative impact of high PAR signals.  This is because the randomized measurement procedure of CS will produce measurements having a PAR that is independent of the input signal's PAR.  For high PAR signals, this results in a substantial improvement.

\section{Simulations}
\label{sec:tests}

Hardware devices are currently under construction that will permit laboratory testing of CS-based acquisition receivers~\cite{TroppLDRB_Beyond,SlaviLDB_Compressive}. In the interim, however, we have conducted a set of computer simulations to validate the theoretical results described above.  We begin by recalling that the use of CS implies that the input to the receiver must be reasonably sparse. Thus, our simulations assume that the receiver input consists of no more than $P$ signals, each of bandwidth no greater than $\sparsity/2$ and with a total bandwidth no greater than $\sparsity/2$. We also assume the presence of white additive noise across the input entire band. We assume that the $P$ input components do not overlap in the Fourier domain but otherwise might appear anywhere within the instantaneous bandwidth of the receiver $\ambient/2$.

The simulation testing suite assembled for this effort permits from $1$ to $P$ of these ``voice-like'' signals to be modulated, if desired, translated up to arbitrary frequencies, summed, corrupted with additive white noise, and then applied to the software-based emulation of a CS receiver. The quality of the recovered individual components is then quantified using the RSNR as defined in~\eqref{eq:RSNR}.

Figure~\ref{fig:sims1}(a) shows the results of a set of simulations of a CS-based wideband signal acquisition system.  In this simulation the signal to be acquired consists of a single 3.1 kHz-wide unmodulated voice signal single-side-band-upconverted to a frequency within the 1 MHz input bandwidth of the receiver. Performance is measured as a function of the subsampling factor $\subsamp$. The experiment shown in Fig.~\ref{fig:sims1}(a) is conducted at three ISNRs --- $60$, $40$, and $20$dB --- where ISNR is simply the ratio of the signal power to that of the noise within the 3.1 kHz bandwidth occupied by the signal. The impact of this signal noise on the reconstruction error is quantified using the RSNR and is evaluated for three sampling/recovery strategies:
\begin{itemize}

\item Bandpass sampling --- This is not a recommended practical technique, but it serves as a benchmark since, like CS, it does not exploit any prior knowledge as to where the signal resides in the input spectrum. It is important to recall that this method ``folds'' the input spectrum so that signal frequencies can no longer be unambiguously determined at the receiver.

\item Oracle-assisted signal recovery from compressive measurements --- Although typically not practical, again, the oracle provides a way to determine what portion of any observed received quality degradation is completely unavoidable within the CS framework and what portion is due to the recovery algorithm's inability to determine the correct spectral support.

\item Practical CS-based signal recovery using CoSaMP~\cite{NeedeT_CoSaMP} to determine the support of the input signal.

\end{itemize}

We can make the following observations from the experimental results depicted in Fig.~\ref{fig:sims1}(a):
\begin{itemize}
\item For small amounts of subsampling the RSNR of both the bandpass sampled signal and the oracle-assisted CS recovery is degraded at a rate of $3$dB for each octave increase in the ratio $\subsamp$, exactly as predicted by Theorem~\ref{thm:noisefolding}.

\item The RSNR of the oracle-assisted recovery approach closely follows the bandpass sampling RSNR across the entire range considered for $\subsamp$.  The performance of the CoSaMP algorithm generally tracks the others until $\subsamp$ nears the theoretical limit:
\begin{equation} \label{eq:cslimit}
\subsampcs = \kappa_0 \frac{\subsampmax}{\log \subsampmax} = \kappa_0 \frac{\ambient/\sparsity}{\log(\ambient/\sparsity)}.
\end{equation}
Note that for these experiments, $\subsampmax = (2 \cdot 10^6) / (3.1 \cdot 10^3) \approx 645$, and thus for $\kappa_0 = \frac{1}{2}$ we have $\log_2(\subsampcs) \approx 5.6$.  In Fig.~\ref{fig:sims1}(a) we observe that we do not begin to observe a dramatic difference between the performance of oracle-aided CS and CoSaMP until $\log_2(\subsamp) > 5$. This demonstrates that the practical impact of $\kappa_0$ on the maximum subsampling factor $\subsampcs$ is likely to be minimal.

\item The RSNR performance of the CoSaMP algorithm generally tracks the others, but performs progressively more poorly for high subsampling factors.  Moreover, its performance collapses as the theoretical limit (\ref{eq:cslimit}) is reached. 

\item The oracle-aided CS algorithm performance closely matches the bandpass sampling RSNR, even for values of $\rho$ where CoSaMP fails.  This suggests that CoSaMP is unable to identify the correct locations of the nonzero Fourier coefficients, since this is the only difference between CoSaMP and the oracle-aided algorithm.  Thus, if any side information concerning the locations of these nonzeros were available (as in a streaming context, e.g., see~\cite{Vaswa_Analyzing}), then one could expect that exploiting this information would have a significant impact on the RSNR.

\end{itemize}

\begin{figure*}[t]
\centering
\hspace{-.075\linewidth}
\begin{minipage}{.45\linewidth}
\centering
\begin{tabular}{cc}
\raisebox{30mm}{\small{(a)}} &\includegraphics[width=.82\imgwidth]{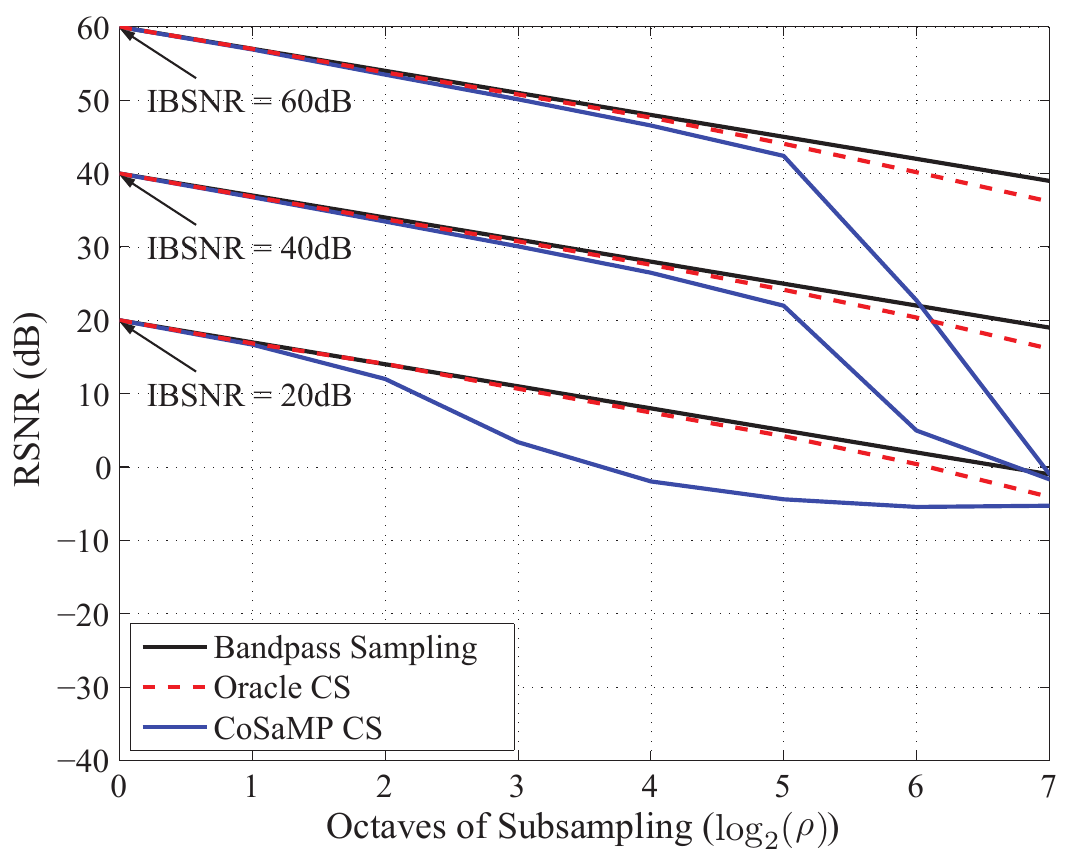} \\
\end{tabular}
\end{minipage}
\hspace{.05\linewidth}
\begin{minipage}{.45\linewidth}
\centering
\begin{tabular}{cc}
\raisebox{30mm}{\small{(b)}} &\includegraphics[width=.82\imgwidth]{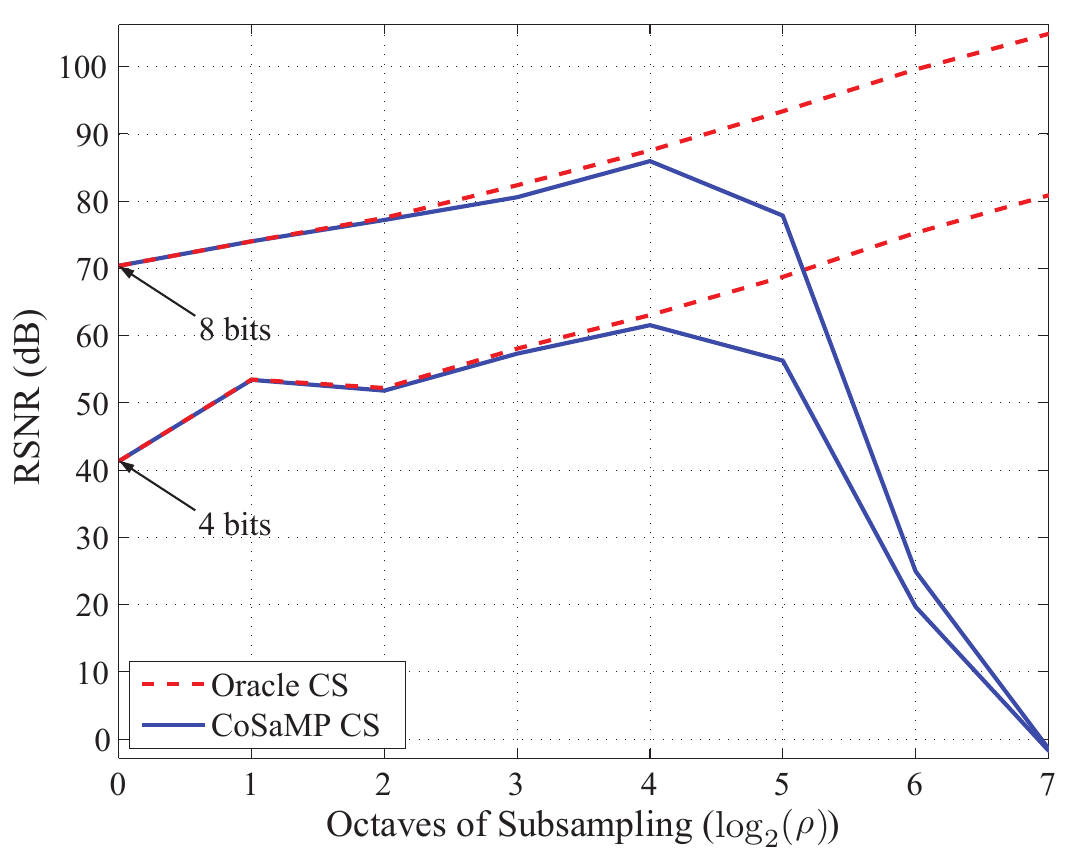} \\
\end{tabular}
\end{minipage}
\caption{Impact of noise and quantization on RSNR as a function of receiver subsampling ratio $\subsamp$. (a) shows the RSNR for an environment consisting of a single unmodulated voice channel in the presence of additive white noise.  The bandpass sampling curve shows the SNR degradation attributable to the $3$dB per octave SNR loss. (b) shows the RSNR for an environment consisting of a noise-free single unmodulated voice channel and quantized measurements starting at a bit-depth of $4$ or $8$ bits per measurement when $\log_{2}(\rho) = 0$.  We increased the bit-depth as a function of the sample rate according to the trends outlined in~\cite{LeRonRee::2005::Analog-to-Digital-Converters}.  We see a marked improvement in RSNR as a direct result of the sampling rate being decreased. \label{fig:sims1}}
\end{figure*}

We next conduct an experiment that demonstrates how the dynamic range of a CS system can be increased, depicted in Fig.~\ref{fig:sims1}(b). As noted earlier, any improvement in the SQNR of the CS measurements will translate to an improved dynamic range.  Thus, in our experiments, we compute the average RSNR obtained after recovery from quantized CS measurements as a proxy for the dynamic range.  Furthermore, we make use of the trends outlined in \cite{LeRonRee::2005::Analog-to-Digital-Converters} that show that the number of bits per measurement grows according to $b = \lambda - 10\log10(B/\rho)/2.3$, where $\lambda$ is a constant determined by the bit-depth of a Nyquist-rate sampler.  The number of bits per measurements then grows linearly with the octaves of subsampling, with slope of about $1.3$.  This relationship between sample rate and bit-depth is fundamental to understanding the dynamic range benefits of CS systems.

Specifically, in each trial we generate a single voice-like signal and compute measurements with the CS receiver using the same setup as in the previous experiment but where the signal is noise-free so that we can isolate the impact of quantization noise.  The measurements are scaled to make use of the full range of the quantizer, quantized to $b$ bits, and then recovered using CoSaMP (solid line) and the oracle recovery algorithm (dashed line).  As explained earlier, for a fixed power and cost, decreasing the sampling rate enables us to choose a higher bit-rate quantizer.  To reflect this, we examine two cases: $\lambda$ such that the Nyquist-rate sampler starts at (\emph{i}) $b=4$ bits and (\emph{ii}) $b=8$ bits.

From this experiment we see that in both cases, the RSNR grows significantly, achieving a $20$dB gain at $4$ octaves of subsampling over Nyquist sampling in  the $4$ bit case and a $17$dB gain in the $8$ bit case.  The performance then decreases as we move to a regime where CS recovery is no longer sustainable.  The oracle performance continues to improve as subsampling is further increased.  This experiment highlights the very real benefit of reduced sampling rates; easing the sampling rate requirement can allow us to use higher fidelity hardware components, such as high bit-depth quantizers.

While it appears that there is a straightforward tradeoff between noise folding (bad) and dynamic range (good) in CS systems, the reality is more nuanced.  For instance, when a both signal and quantization noise are present, it is not always possible to mitigate the negative effects of noise folding by increasing the ADC bit-depth, since there may be little benefit to increasing the precision of already noisy measurements. See~\cite{LaskaB_Regime} for further discussion of this issue.

\section{Using the Design Rules to Evaluate a CS Receiver}
\label{sec:sims}

The simulations presented in Section~\ref{sec:tests} provide an initial validation of the engineering design relationships established in this paper. To see how they might be used in practice, we return briefly to the example set of system requirements shown in Table~\ref{tab:reqs}.  By applying the rules described in Section~\ref{sec:csthy} and validated in~\ref{sec:tests} we find that $\subsampcs$, the maximum subsampling factor in a CS system, is about $160$. This implies sampling rate, and data link transmission rate, can be reduced from $1$ GHz to $6.25$ MHz but at a noise floor loss of about $22$dB. Use of the formulae in Section~\ref{sec:tests} indicates, however, that an improvement of $9$--$10$ bits can be achieved owing to the lower sampling rate. If we presume that the ADC used at a $1$ GHz sampling rate might have $8$ bits of dynamic range, then the compressive sensing receiver should be able to achieve $17$ or more, leading to a system dynamic range of greater than $100$dB.  Comparing these results with the objectives in Table~\ref{tab:reqs} shows the remarkable result that a CS-based acquisition system can theoretically meet the very stringent  and rarely attained instantaneous bandwidth and dynamic range requirements, but at the cost of a reduced RSNR.

\section{Conclusions, Implications, and Recommendations for Future Work}
\label{sec:conc}

This paper has examined how CS can provide an exciting new degree of freedom in the design of high-performance signal acquisition systems. Specifically, the results reported in this paper can be captured succinctly as follows:
\begin{itemize}
\item The application of CS theory to the problem of designing an RF receiver indicates that the approach is indeed feasible, and that it should reduce the size, weight, power consumption and monetary cost of the receiver.  This comes at a cost of an increased noise figure and an increased amount of computation required at the downstream ``processing center.''
\item There is a direct, predictable relationship between the subsampling factor and the noise figure of the receiver. Specifically, we lose $3$dB of RNSR for each halving of the sampling rate.
\item Simulation results indicate that a properly designed CS system can approach, and even meet, the theoretically predicted performance, if the ISNR is high enough.
\item Since it permits the use of lower-rate, but higher performance ADCs, the introduction of CS can substantially improve the dynamic range of a receiver system.
\end{itemize}

These results mean that CS introduces new tradeoffs in the design of signal acquisition systems.  While a poorer noise figure reduces the sensitivity of a receiver, at the ``systems level'' that might be acceptable in trade for what one gets for it --- much wider instantaneous bandwidth, improved dynamic range, and reduction of virtually all elements of the SWAP ``cost vector'' at the sensor, where it usually matters the most. We also note that the decimation permitted by using CS in a sparse signal environment permits a significant dynamic range improvement --- up to $20$dB in the settings considered in this paper. Examples of how this tradeoff can be exploited in practical systems are explored in~\cite{futureDASP}.

Thus we conclude that further investigation in this area will produce both theoretical and practical fruit. There are two areas in which we recommend immediate emphasis ---  (\emph{i}) verification that CS receivers can be physically implemented with performance we have theoretically predicted, and (\emph{ii}) more work on practical and efficient  \emph{processing center} algorithms for signal reconstruction, or, equivalently, parameter estimation (e.g., emitter location) from the incoming compressive measurements.  Successes in these two areas will make CS an important tool in the toolbox of radio system designers.

\section*{Appendix}

We now provide proofs of the more technical results in this paper.  In several of the lemmas we use the notation $\lambda_j(\A)$ to denote the $j^{\mathrm{th}}$ largest eigenvalue of $\A$.  We also let $s_j(\A)$ denote the $j^\mathrm{th}$ singular value of $\A$, i.e., $s_j(\A) = \lambda_j(\A^T \A)$. Before establishing our main result concerning oracle-assisted recovery, we first state the following useful lemma.  The proof of this lemma can be found in~\cite{Daven_Random}.
\begin{lemma}[Lemma 7.1 of~\cite{Daven_Random}] \label{lem:oraclebound}
Suppose that $\R$ is an $\meas \times \ambient$ matrix and let $\Lambda$ be a set of indices with $|\Lambda| \le \sparsity$. If $\R$ satisfies the RIP of order $\sparsity$ with constant $\delta$, then for $j=1,2,\ldots, \sparsity$ we have
\begin{equation} \label{eq:svdbound}
\frac{1}{\sqrt{1+\delta}} \le s_j(\R_{\Lambda}^\dag) \le \frac{1}{\sqrt{1-\delta}}.
\end{equation}
\end{lemma}

\begin{proof}[Proof of Theorem~\ref{thm:expectederror}]
Recall that for the oracle-assisted recovery algorithm, since the RIP ensures that $\R_\Lambda$ is full rank, we have that
$$
\widehat{\balpha}|_{\Lambda} = \balpha|_{\Lambda} + \R_{\Lambda}^\dag \e.
$$
Thus, our goal is to estimate $\expec{ \norm{\R_{\Lambda}^\dag \e }^2}$.  Towards this end, we first note that for any $\A$, we have
\begin{align*}
\expec{\norm{\A \e}^2} & = \expec{ \trace{\A \e (\A \e)^T} } \\
 & = \trace{ \A \expec{ \e \e^T} \A^T }.
\end{align*}
Since $\e$ is a white random vector, this reduces to
\begin{align*}
\expec{\norm{\A \e}^2} & = \trace{ \A \left( \sigma_{\e}^2 \I_{\meas} \right) \A^T } \\
& = \sigma_{\e}^2 \trace{\A \A^T} = \sigma_{\e}^2 \norm[F]{\A}^2,
\end{align*}
where $\norm[F]{\cdot}$ denotes the Frobenius norm of $\A$.  Next we recall that the Frobenius norm of a $\sparsity \times \meas$ matrix with $\sparsity < \meas$ can also be calculated as
$$
\norm[F]{\A}^2 = \sum_{j=1}^{\sparsity} s_j(\A)^2.
$$
Thus, we have
\begin{equation} \label{eq:recerrorsvs}
\expec{\norm{\R_{\Lambda}^\dag \e}^2} = \sigma_{\e}^2 \sum_{j=1}^{\sparsity} s_j(\R_{\Lambda}^\dag)^2.
\end{equation}
From Lemma~\ref{lem:oraclebound}, $s_j(\R_{\Lambda}^\dag) \in \left[1/\sqrt{1+\delta},1/\sqrt{1-\delta} \right]$ for $j=1,2,\ldots,\sparsity$, and hence
$$
\frac{\sparsity}{1+\delta} \le \sum_{j=1}^{\sparsity} s_j\left( \R_{\Lambda}^\dag \right)^2 \le \frac{\sparsity}{1-\delta},
$$
which combined with (\ref{eq:recerrorsvs}) establishes (\ref{eq:thm:expectederror_a}).  Next, we note that
$$
\frac{\mathrm{RSNR}}{\mathrm{MSNR}} = \frac{\norm{\balpha}^2}{\norm{\R \balpha}^2}  \frac{\expec{\norm{ \e }^2}}{\expec{\norm{ \widehat{\balpha} - \balpha }^2}}  = \frac{\norm{\balpha}^2}{\norm{\R \balpha}^2}   \frac{\meas \sigma_{\e}^2}{\expec{\norm{ \widehat{\balpha} - \balpha }^2}}.
$$
Combining this with the RIP and with~\eqref{eq:thm:expectederror_a} yields~\eqref{eq:thm:expectederror_b}.
\end{proof}

We next establish the following preliminary lemma which we will use to establish Lemma~\ref{lem:RIP_orthoprojector}.
\begin{lemma} \label{lem:eigbound}
Let $\A$ be a real, symmetric, positive definite $\ambient \times \ambient$ matrix, and let $\B$ be a real $\ambient \times \meas$ matrix with $\meas \le \ambient$.  Then
\begin{equation} \label{eq:lambdamax}
\lambda_{\mathrm{max}}(\B^T \A \B) \le \lambda_{\mathrm{max}}(\A) \lambda_{\mathrm{max}}(\B^T\B)
\end{equation}
and
\begin{equation} \label{eq:lambdamin}
\lambda_{\mathrm{min}}(\B^T \A \B) \ge \lambda_{\mathrm{min}}(\A) \lambda_{\mathrm{min}}(\B^T \B).
\end{equation}
\end{lemma}
\begin{proof}
Recall that since both $\A$ and $\B^T \B$ are real, symmetric matrices, we have that
$$
\lambda_{\mathrm{min}}(\A) \le \frac{\x^T \A \x}{\x^T \x} \le \lambda_{\mathrm{max}}(\A)
$$
for all $\x \in \real^\ambient$ and
$$
\lambda_{\mathrm{min}}(\B^T \B) \le \frac{\y^T \B^T \B \y}{\y^T \y} \le \lambda_{\mathrm{max}}(\B^T \B)
$$
for all $\y \in \real^\meas$. Thus, by setting $\x = \B \y$, we have
\begin{align*}
\frac{\y^T \B^T \A \B \y}{\y^T \y} = \frac{\x^T \A \x}{\y^T \y} & \le \lambda_{\mathrm{max}}(\A) \frac{\x^T \x}{\y^T \y} \\
& = \lambda_{\mathrm{max}}(\A) \frac{\y^T \B^T \B \y}{\y^T \y} \\
&  \le \lambda_{\mathrm{max}}(\A) \lambda_{\mathrm{max}}(\B^T\B),
\end{align*}
which establishes (\ref{eq:lambdamax}).  The proof of (\ref{eq:lambdamin}) follows from a similar calculation.
\end{proof}

\begin{proof}[Proof of Lemma~\ref{lem:RIP_orthoprojector}]
To begin, let $\R = \U \bSigma \V^T$ denote the reduced form of the singular value decomposition of $\R$, so that if $\R$ has rank $\meas$ then $\U$ is an $\meas \times \meas$ unitary matrix, $\bSigma$ is a diagonal $\meas \times \meas$ matrix whose entries are given by the singular values $s_j(\R)$, and $\V$ is an $\ambient \times \meas$ matrix with orthonormal columns.  Let
$$
\widetilde{\R} = \bSigma^{-1} \U^T \R = \bSigma^{-1} \U^T \U \bSigma \V^T = \V^T
$$
and observe that from the properties of $\V$ and the singular value decomposition we have that $\widetilde{\R}$ has orthonormal rows and the same row space as $\R$. Thus, it remains to show (\ref{eq:RIP_ortho_bound}).  Towards this end, we note that if $\balpha$ is $\sparsity$-sparse, then (\ref{eq:RIP_ortho_bound}) is equivalent to requiring that
$$
\frac{1 - \delta}{s_{\mathrm{max}}^2(\R)} \le \frac{\norm{\widetilde{\R}_\Lambda \x}^2}{\norm{\x}^2} \le \frac{1 + \delta}{s_{\mathrm{min}}^2(\R)}
$$
holds for all $\x \in \real^{|\Lambda|}$ and all $\Lambda$ with $|\Lambda| \le \sparsity$.  By noting that $\widetilde{\R}_{\Lambda} = \bSigma^{-1} \U^T \R_{\Lambda}$ and applying Lemma~\ref{lem:eigbound}, we have that
\begin{align*}
\frac{\norm{\widetilde{\R}_\Lambda \x}^2}{\norm{\x}^2} = \frac{ \x^T \widetilde{\R}_\Lambda^T \widetilde{\R}_\Lambda \x}{\x^T \x} & \le \lambda_{\mathrm{max}}(\widetilde{\R}_\Lambda^T \widetilde{\R}_\Lambda) \\
& = \lambda_{\mathrm{max}}(\R_\Lambda^T \U \bSigma^{-1} \bSigma^{-1} \U^T \R_\Lambda)   \\
& \le \lambda_{\mathrm{max}}(\bSigma^{-2}) \lambda_{\mathrm{max}}(\R_\Lambda^T \U \U^T \R_\Lambda) \\
& = \frac{1}{s_{\mathrm{min}}^2(\R)} \lambda_{\mathrm{max}}(\R_\Lambda^T \R_\Lambda)   \\
& \le \frac{1 + \delta}{s_{\mathrm{min}}^2(\R)}.
\end{align*}
The lower bound follows via a similar argument.
\end{proof}

\begin{proof}[Proof of Theorem~\ref{thm:noiseexpansion}]
We begin by noting that $\expec{\R \n} = \R \expec{\n} = 0$, so that $\R \n$ is zero-mean, as desired.  Hence, we now consider $\expec{\R \n (\R \n)^T}$.  Note that the assumptions on $\R$ imply that $\R \R^T = \subsamp \I_{\meas}$, and hence
\begin{align}
\expec{\R \n (\R \n)^T} & = \R \expec{ \n \n^T} \R^T \\
&  = \sigma_{\n}^2 \R \R^T = \subsamp \sigma_{\n}^2 \I_{\meas},\label{eq:noiseexpansion2}
\end{align}
which establishes that $\sigma_{\R \n}^2 = \subsamp \sigma_{\n}^2$. Furthermore, since
\begin{align*}
\frac{\mathrm{MSNR}}{\mathrm{ISNR}} & = \frac{\norm{\R \balpha}^2}{\norm{\balpha}^2} \frac{\expec{\norm{\n|_\Lambda}^2}}{\expec{\norm{ \R \n }^2}} \\
& = \frac{\norm{\R \balpha}^2}{\norm{\balpha}^2} \frac{\sparsity \sigma_{\n}^2}{\meas \sigma_{\R \n}^2} = \frac{\norm{\R \balpha}^2}{\norm{\balpha}^2} \frac{\sparsity}{\ambient},
\end{align*}
the bound in (\ref{eq:noiseexpansionb}) simply follows from the RIP.
\end{proof}

\begin{proof}[Proof of Lemma~\ref{lem:betaopt}]
\label{proof:betaopt}
Begin by taking $\beta = G/\norm[\infty]{\x}$.  Observe that
$$
\norm[\infty]{\beta \x} = \beta \norm[\infty]{\x} = G,
$$
and thus no entries of $\beta \x$ exceed the saturation level $G$.  Hence, we can bound the quantization error as
\begin{equation} \label{eq:SQNRlem1}
\norm{\beta \x - \quant{\beta \x}}^2 \le \ambient \left( \frac{\Delta}{2} \right)^2.
\end{equation}
We also have that
\begin{equation} \label{eq:SQNRlem2}
\norm{\beta \x}^2 = \beta^2 \norm{\x}^2 = \frac{G^2 \norm{\x}^2}{\norm[\infty]{\x}^2}.
\end{equation}
Combining (\ref{eq:SQNRlem1}) and (\ref{eq:SQNRlem2}), we obtain that
$$
\mathrm{SQNR}(\beta \x) = \frac{\norm{\beta \x}^2}{\norm{\beta \x - \quant{\beta \x}}^2} \ge \frac{ G^2 \norm{\x}^2/\norm[\infty]{\x}^2}{\ambient \left( \Delta/2 \right)^2},
$$
which simplifies to yield the desired result.
\end{proof}

\begin{proof}[Proof of Theorem~\ref{thm:dr}]
We begin by considering $\betamin(\x)$.  Recall that for all scalings $\beta < G/\norm[\infty]{\x}$ we have that $\norm[\infty]{\beta \x} < G$, so that there are no saturations.  Thus we can bound the SQNR as
$$
\mathrm{SQNR}(\beta \x) \ge \frac{\beta^2 \norm{\x}^2}{\ambient (\Delta/2)^2}.
$$
Thus, if we ensure that
$$
\frac{\beta^2 \norm{\x}^2}{\ambient (\Delta/2)^2} \geq C
$$
then we also guarantee that $\mathrm{SQNR}(\beta \x) \geq C$.  This will occur provided that
$$
\beta^2 \ge \frac{C\ambient}{\norm{\x}^2}  (\Delta/2)^2,
$$
and thus we can set
$$
\betamin(\x)^2 = \frac{C\ambient}{\norm{\x}^2}  (\Delta/2)^2.
$$

We now turn to $\betamax(\x)$.  Since we are now considering $\beta > G/\norm[\infty]{\x}$, there will be at least one entry of $\x$ that takes a value greater than $G$ and thus saturates. Furthermore, the saturated value is guaranteed to have error greater than $\Delta/2$ since our quantizer represents a maximum value of $G - \Delta/2$.   Thus, we observe that the total quantization error is less than the error of a signal where each element takes the value of the maximum saturated measurement.  If we define $\overline{G} = G - \Delta/2$ then we have that
\begin{equation} \label{eq:dr1}
\mathrm{SQNR}(\beta \x) \ge \frac{\beta^2 \norm{\x}^2}{\ambient(\beta \|\x\|_{\infty} - \overline{G})^{2}}.
\end{equation}

By design, we have that $\beta\|\x\|_{\infty} > G$, and hence
\begin{align*}
\left(  \beta \|\x\|_{\infty}- \overline{G} \right)^2 & = \beta^2 \|\x\|_{\infty}^2 - 2 \overline{G} \beta \|\x\|_{\infty} + \overline{G}^2    \\
& \le \beta^2 \|\x\|_{\infty}^{2} - 2 \overline{G} G + \overline{G}^2 \\
& = \beta^2 \|\x\|_{\infty}^2 - G^{2} + (\Delta/2)^2.
\end{align*}

From this we observe that
$$
\frac{\beta^2 \norm{\x}^2}{\ambient(\beta \|\x\|_{\infty} - \overline{G})^{2}} \ge \frac{\beta^2 \norm{\x}^2}{\ambient(\beta^2 \|\x\|_{\infty}^2 - G^{2} + (\Delta/2)^2)},
$$
and so from (\ref{eq:dr1}) we have that if
$$
\frac{\beta^2 \norm{\x}^2}{\ambient(\beta^2 \|\x\|_{\infty}^2 - G^{2} + (\Delta/2)^2)} > C
$$
then $\mathrm{SQNR}(\beta \x) > C$.  By rearranging, we see that this will occur provided that
$$
\beta^2 < \frac{C\ambient}{\|\x\|_{2}^{2}}\left( \frac{G^{2}-(\Delta/2)^{2}}{C \gamma(\x)^2 - 1} \right).
$$
Thus we can set
$$
\betamax(\x)^2 = \frac{C\ambient}{\|\x\|_{2}^{2}}\left( \frac{G^{2}-(\Delta/2)^{2}}{C \gamma(\x)^2 - 1} \right).
$$

Combining our expressions for $\betamin(\x)$ and $\betamax(\x)$ we obtain
\begin{equation*} \label{eq:finaldr}
\mathrm{DR}_{C}(\x) \ge \left( \frac{\betamax(\x)}{\betamin(\x)} \right)^2 = \frac{1}{C\gamma(\x)^2 - 1} \left( \left(\frac{2G}{\Delta}\right)^{2} -1 \right),
\end{equation*}
which simplifies to establish (\ref{eq:DR1}).
\end{proof}

\bibliographystyle{IEEEtran}
\small
\bibliography{bibpreamble,bibmain}

\begin{thebibliography}{10}
\providecommand{\url}[1]{#1}
\csname url@samestyle\endcsname
\providecommand{\newblock}{\relax}
\providecommand{\bibinfo}[2]{#2}
\providecommand{\BIBentrySTDinterwordspacing}{\spaceskip=0pt\relax}
\providecommand{\BIBentryALTinterwordstretchfactor}{4}
\providecommand{\BIBentryALTinterwordspacing}{\spaceskip=\fontdimen2\font plus
\BIBentryALTinterwordstretchfactor\fontdimen3\font minus
  \fontdimen4\font\relax}
\providecommand{\BIBforeignlanguage}[2]{{%
\expandafter\ifx\csname l@#1\endcsname\relax
\typeout{** WARNING: IEEEtran.bst: No hyphenation pattern has been}%
\typeout{** loaded for the language `#1'. Using the pattern for}%
\typeout{** the default language instead.}%
\else
\language=\csname l@#1\endcsname
\fi
#2}}
\providecommand{\BIBdecl}{\relax}
\BIBdecl

\bibitem{TreicDB_Application}
J.~Treichler, M.~Davenport, and R.~Baraniuk, ``Application of compressive
  sensing to the design of wideband signal acquisition receivers,'' in
  \emph{Proc. Defense Apps. of Signal Processing (DASP)}, Lihue, Hawaii, Sept.
  2009.

\bibitem{Baran_Compressive}
R.~Baraniuk, ``Compressive sensing,'' \emph{IEEE Signal Processing Mag.},
  vol.~24, no.~4, pp. 118--120, 124, 2007.

\bibitem{Cande_Compressive}
E.~Cand\`{e}s, ``Compressive sampling,'' in \emph{Proc. Int. Congress of
  Math.}, Madrid, Spain, Aug. 2006.

\bibitem{DavenDEK_Introduction}
M.~Davenport, M.~Duarte, Y.~Eldar, and G.~Kutyniok, ``Introduction to
  compressed sensing,'' in \emph{Compressed Sensing: Theory and
  Applications}.\hskip 1em plus 0.5em minus 0.4em\relax Cambridge University
  Press, 2012.

\bibitem{Donoh_Compressed}
D.~Donoho, ``Compressed sensing,'' \emph{IEEE Trans. Inform. Theory}, vol.~52,
  no.~4, pp. 1289--1306, 2006.

\bibitem{TroppLDRB_Beyond}
J.~Tropp, J.~Laska, M.~Duarte, J.~Romberg, and R.~Baraniuk, ``Beyond {N}yquist:
  {E}fficient sampling of sparse, bandlimited signals,'' \emph{IEEE Trans.
  Inform. Theory}, vol.~56, no.~1, pp. 520--544, 2010.

\bibitem{DuartB_Spectral}
M.~Duarte and R.~Baraniuk, ``Spectral compressive sensing,'' 2011, {P}reprint.

\bibitem{DavenW_Compressive}
M.~Davenport and M.~Wakin, ``Compressive sensing of analog signals using
  discrete prolate spheroidal sequences,'' \emph{To appear in {\em Appl.
  Comput. Harmon. Anal.}}, 2012.

\bibitem{CandeT_Decoding}
E.~Cand\`{e}s and T.~Tao, ``Decoding by linear programming,'' \emph{IEEE Trans.
  Inform. Theory}, vol.~51, no.~12, pp. 4203--4215, 2005.

\bibitem{BaranDDW_Simple}
R.~Baraniuk, M.~Davenport, R.~DeVore, and M.~Wakin, ``{A simple proof of the
  restricted isometry property for random matrices},'' \emph{Const. Approx.},
  vol.~28, no.~3, pp. 253--263, 2008.

\bibitem{TroppWDBB_Random}
J.~Tropp, M.~Wakin, M.~Duarte, D.~Baron, and R.~Baraniuk, ``Random filters for
  compressive sampling and reconstruction,'' in \emph{Proc. IEEE Int. Conf.
  Acoust., Speech, and Signal Processing (ICASSP)}, Toulouse, France, May 2006.

\bibitem{LaskaSB_Polyphase}
J.~Laska, J.~Slavinsky, and R.~Baraniuk, ``The polyphase random demodulator for
  wideband compressive sensing,'' in \emph{Proc. Asilomar Conf. Signals,
  Systems, and Computers}, Pacific Grove, CA, Nov. 2011.

\bibitem{MishaE_From}
M.~Mishali and Y.~C. Eldar, ``From theory to practice: {S}ub-{N}yquist sampling
  of sparse wideband analog signals,'' \emph{IEEE J. Select. Top. Signal
  Processing}, vol.~4, no.~2, pp. 375--391, 2010.

\bibitem{BajwaHRWN_Toeplitz}
W.~Bajwa, J.~Haupt, G.~Raz, S.~Wright, and R.~Nowak, ``Toeplitz-structured
  compressed sensing matrices,'' in \emph{Proc. IEEE Work. Stat. Signal
  Processing}, Madison, WI, Aug. 2007.

\bibitem{Rombe_Compressive}
J.~Romberg, ``Compressive sensing by random convolution,'' \emph{SIAM J. Imag.
  Sci.}, vol.~2, no.~4, pp. 1098--1128, 2009.

\bibitem{SlaviLDB_Compressive}
J.~P. Slavinsky, J.~Laska, M.~Davenport, and R.~Baraniuk, ``The compressive
  mutliplexer for multi-channel compressive sensing,'' in \emph{Proc. IEEE Int.
  Conf. Acoust., Speech, and Signal Processing (ICASSP)}, Prague, Czech
  Republic, May 2011.

\bibitem{YuHS_Mixed-Signal}
Z.~Yu, S.~Hoyos, and B.~Sadler, ``Mixed-signal parallel compressed sensing and
  reception for cognitive radio,'' in \emph{Proc. IEEE Int. Conf. Acoust.,
  Speech, and Signal Processing (ICASSP)}, Las Vegas, NV, Apr. 2008.

\bibitem{LexaDT_Reconciling}
M.~Lexa, M.~Davies, and J.~Thompson, ``Reconciling compressive sampling schemes
  for spectrally-sparse continuous-time signals,'' \emph{IEEE Trans. Signal
  Processing}, vol.~60, no.~1, pp. 155--171, 2012.

\bibitem{CandeRT_Stable}
E.~Cand\`{e}s, J.~Romberg, and T.~Tao, ``Stable signal recovery from incomplete
  and inaccurate measurements,'' \emph{Comm. Pure Appl. Math.}, vol.~59, no.~8,
  pp. 1207--1223, 2006.

\bibitem{NeedeT_CoSaMP}
D.~Needell and J.~Tropp, ``{CoSaMP}: {I}terative signal recovery from
  incomplete and inaccurate samples,'' \emph{Appl. Comput. Harmon. Anal.},
  vol.~26, no.~3, pp. 301--321, 2009.

\bibitem{BlumeD_Iterative}
T.~Blumensath and M.~Davies, ``Iterative hard thresholding for compressive
  sensing,'' \emph{Appl. Comput. Harmon. Anal.}, vol.~27, no.~3, pp. 265--274,
  2009.

\bibitem{CohenDD_Instance}
A.~Cohen, W.~Dahmen, and R.~DeVore, ``Instance optimal decoding by thresholding
  in compressed sensing,'' in \emph{Int. Conf. Harmonic Analysis and Partial
  Differential Equations}, Madrid, Spain, Jun. 2008.

\bibitem{CandeT_Dantzig}
E.~Cand\`{e}s and T.~Tao, ``The {D}antzig selector: {S}tatistical estimation
  when $p$ is much larger than $n$,'' \emph{Ann. Stat.}, vol.~35, no.~6, pp.
  2313--2351, 2007.

\bibitem{Daven_Random}
M.~Davenport, ``Random observations on random observations: {S}parse signal
  acquisition and processing,'' Ph.D. dissertation, Rice University, Aug. 2010.

\bibitem{AeroSZ_Information}
S.~Aeron, V.~Saligrama, and M.~Zhao, ``Information theoretic bounds for
  compressed sensing,'' \emph{IEEE Trans. Inform. Theory}, vol.~56, no.~10, pp.
  5111--5130, 2010.

\bibitem{CastroYoninaNoiseFold}
E.~Arias-{C}astro and Y.~Eldar, ``Noise folding in compressed sensing,''
  \emph{IEEE Signal Processing Lett.}, vol.~18, no.~8, pp. 478--481, 2011.

\bibitem{VaughSW_Theory}
R.~Vaughan, N.~Scott, and R.~White, ``The theory of bandpass sampling,''
  \emph{IEEE Trans. Signal Processing}, vol.~39, no.~9, pp. 1973--1984, 1991.

\bibitem{RaskuWY_Minimax}
G.~Raskutti, M.~Wainwright, and B.~Yu, ``Minimax rates of estimation for
  high-dimensional linear regression over $\ell_q$ balls,'' \emph{IEEE Trans.
  Inform. Theory}, vol.~57, no.~10, pp. 6976--6994, 2011.

\bibitem{CandeD_How}
E.~Cand{\`e}s and M.~Davenport, ``How well can we estimate a sparse vector?''
  2011, {P}reprint.

\bibitem{CastrCD_Fundamental}
E.~Arias-{C}astro, E.~Cand\`{e}s, and M.~Davenport, ``On the fundamental limits
  of adaptive sensing,'' 2011, {P}reprint.

\bibitem{DaviesGuo_sample}
M.~Davies and C.~Guo, ``Sample-distortion functions for compressed sensing,''
  in \emph{Proc. Allerton Conf. Communication, Control, and Computing},
  Monticello, IL, 2011.

\bibitem{LeRonRee::2005::Analog-to-Digital-Converters}
B.~Le, T.~W. Rondeau, J.~H. Reed, and C.~W. Bostian, ``Analog-to-digital
  converters,'' \emph{IEEE Sig. Proc. Mag.}, Nov. 2005.

\bibitem{Lyons_Understanding}
R.~Lyons, \emph{{Understanding Digital Signal Processing}}.\hskip 1em plus
  0.5em minus 0.4em\relax Upper Saddle River, NJ: Pearson Education, Inc.,
  2004.

\bibitem{LaskaBDB_Democracy}
J.~Laska, P.~Boufounos, M.~Davenport, and R.~Baraniuk, ``Democracy in action:
  {Q}uantization, saturation, and compressive sensing,'' \emph{Appl. Comput.
  Harmon. Anal.}, vol.~31, no.~3, pp. 429--443, 2011.

\bibitem{Vaswa_Analyzing}
N.~Vaswani, ``Analyzing least squares and {K}alman filtered compressed
  sensing,'' in \emph{Proc. IEEE Int. Conf. Acoust., Speech, and Signal
  Processing (ICASSP)}, Taipei, Taiwan, Apr. 2009.

\bibitem{LaskaB_Regime}
J.~Laska and R.~Baraniuk, ``Regime change: Bit-depth versus measurement-rate in
  compressive sensing,'' \emph{To appear in {\em IEEE Trans. Signal
  Processing}}, 2012.

\bibitem{futureDASP}
J.~Treichler, M.~Davenport, J.~Laska, and R.~Baraniuk, ``Dynamic range and
  compressive sensing acquisition receivers,'' in \emph{Proc. Defense Apps. of
  Signal Processing (DASP)}, Coolum, Australia, Jul. 2011.

\end{thebibliography}

\end{document}